\documentclass[aps,pra,superscriptaddress,showkeys]{revtex4-2}
\usepackage[paperwidth=210mm,paperheight=297mm,centering,hmargin=2cm,vmargin=2cm]{geometry}
\usepackage[american]{babel}
\usepackage{graphicx}
\usepackage{amsmath,amsthm,amsfonts,amssymb}
\usepackage{color}
\usepackage[colorlinks]{hyperref }

\newtheorem{theorem}{Theorem}[section]
\newtheorem{lemma}{Lemma}[section]
\newtheorem{corollary}{Corollary}[section]

\newcommand{\bbN}{\mathbb{N}}
\newcommand{\rmi}{\mathrm{i}}

\newcommand{\di}{\displaystyle}

\begin{document}
\title{The discrete nonlinear Schr\"odinger equation with linear gain and nonlinear loss: the infinite lattice with nonzero boundary conditions and its finite dimensional approximations}
\author{G. Fotopoulos}
\affiliation{General Education Department, Abu Dhabi Polytechnic, P.O.\,Box 111499, Abu Dhabi, UAE}
\email{georgios.fotopoulos@actvet.gov.ae}
\affiliation{Department of Mathematics, Xiamen University Malaysia, 43900, Selangor, Malaysia}
\author{N.\,I. Karachalios}
\affiliation{Department of Mathematics, University of Thessaly, Lamia 35100, Greece}
\email{karan@uth.gr}
\author{V. Koukouloyannis}
\affiliation{Department of Mathematics, University of the Aegean, Karlovasi, 83200 Samos, Greece}
\email{vkouk@aegean.gr}
\author{P. Kyriazopoulos}
\affiliation{Department of Mathematics, University of Thessaly, Lamia 35100, Greece}
\author{K. Vetas}
\affiliation{Department of Mathematics, University of Thessaly, Lamia 35100, Greece}

\begin{abstract}
The study of nonlinear Schr\"odinger-type equations with  nonzero boundary conditions introduces challenging problems both for the continuous (partial differential equation) and  the discrete (lattice) counterparts. They are associated with fascinating dynamics emerging by the ubiquitous phenomenon of modulation instability.  
In this work, we consider the  discrete nonlinear Schr\"odinger equation with linear gain and nonlinear loss. For the infinite lattice supplemented with nonzero boundary conditions, which describe solutions decaying on the top of a finite background, we provide a rigorous proof that for the corresponding initial-boundary value problem, solutions exist for any initial condition, if and only if, the amplitude of the background has a precise value $A_*$ defined by the gain-loss parameters.  We argue that this essential property of this infinite lattice can't be captured by finite lattice approximations of the problem. Commonly, such approximations are provided by lattices with periodic boundary conditions or as it is shown herein, by a modified problem closed with Dirichlet boundary conditions. For the finite dimensional dynamical system defined by the periodic lattice, the dynamics for all initial conditions are captured by a global attractor. Analytical arguments corroborated by numerical simulations show that the global attractor is trivial, defined by a plane wave of amplitude $A_*$. Thus, any instability effects or localized phenomena simulated by the finite system can be only transient prior the convergence to this trivial attractor. Aiming to simulate the dynamics of the infinite lattice as accurately as possible, we study the dynamics of localized initial conditions on the constant background and investigate the potential impact of the global asymptotic stability of the background with amplitude $A_*$ in the long-time evolution of the system. 
\end{abstract}

\keywords{Dissipative discrete nonlinear Schr\"odinger systems, Ablowitz--Ladik lattice, non-zero boundary conditions, modulation instability}
\maketitle

\section{Introduction}
In the present paper, we continue our studies initiated in \cite{GL1} on the  dynamical behavior of the solutions of the following discrete nonlinear Schr\"odinger \cite{KevreDNLS} (DNLS)-type equation 
\begin{equation}
	\label{dnls_gl}
	\rmi\dot{u}_n  + k(u_{n+1} - 2u_n+u_{n-1}) + |u_n|^2 u_n = \rmi\gamma u_n + \rmi\delta |u_n|^2 u_n .
\end{equation} 

Eq.~(\ref{dnls_gl}) is a fundamental nonlinear lattice model
incorporating linear and nonlinear gain/loss effects.
The parameter $\gamma$ describes linear loss ($\gamma<0$) [or gain
($\gamma>0$)],
while $\delta$ describes
nonlinear loss ($\delta<0$) [or gain ($\delta>0$)].
The important of the presence of these
effects is crucial, particularly in the context of nonlinear optics, where the model ~(\ref{dnls_gl}) may describe the evolution of localized modes in optical waveguides, see \cite{Efrem1,Boris1,Boris2,Boris3} (and the references therein).  In this context, $\gamma$ describes a
linear absorption ($\gamma<0$) [or linear amplification ($\gamma>0$)], while $\delta$
stands for nonlinear amplification ($\delta>0$) [or gain saturation ($\delta<0$)].
The importance of these effects in various discrete and continuous set-ups has been also highlighted in \cite{Nail1, KodHas87, Agra1, Agra2, akbook, Gagnon} as they may have a prominent role in the potential stability/instability of
the localized modes. 

In \cite{GL1} we studied the dynamics of the model ~(\ref{dnls_gl}) in the gain/loss regimes for  $\gamma$ and $\delta$, where collapse in finite-time  (or blow-up) is manifested. We identified therein, distinct types of collapse and estimates for the blow-up time which in various cases proved to be sharp.  We remark that the potential destabilization by finite-time collapse has been proved to be a major characteristic of the dynamical behavior of DNLS models incorporating gain/loss effects \cite{Boris2, Boris3}. 

Herein, we will focus on the study of the dynamics in the case of linear gain $\gamma>0$ and nonlinear loss $\delta<0$. While the loss/loss regime $\gamma<0$, $\delta<0$ is yet of physical significance, it is characterized by the decay of solutions for all initial data. Therefore, our primary concern will be the former regime which is characterized by non-trivial long-time asymptotic behavior.

We will consider equation \eqref{dnls_gl} both in infinite and finite and lattices. The first case concerns the system 
supplemented with nonzero boundary conditions at infinity 
\begin{equation}\label{nv}
	\lim_{|n|\rightarrow\infty} u_n(t) =A\exp(\rmi A^2t),\;\;A>0,
\end{equation}
where $A$ is the amplitude of the constant background. Note that the specific choice of  boundary conditions \eqref{nv} is rather general.  The specific choice is imposed by the the unit strength of the nonlinearity in the left-hand side of \eqref{dnls_gl}, which is selected for simplicity. In the case of a general nonlinearity $f(|u_n|^2)u_n$, a consistent problem requires that the right-hand side of the boundary conditions \eqref{nv} must have the form $A\exp[\rmi f(A^2) t]$, for a sufficiently smooth function, so that $\lim_{|n|\rightarrow\infty} f(|u_n|^2)=f(\lim_{|n|\rightarrow\infty}|u_n|^2)=f(A^2)$.

On the one hand, there is a strong motivation from the mathematical point of view for the study of \eqref{dnls_gl} with non-zero boundary conditions like \eqref{nv}. The problem \eqref{dnls_gl}-\eqref{nv} is drastically different from the case of zero boundary conditions and is associated with the emergence of fascinating dynamics and a wide class of localized phenomena which may occur on the top of the finite background $A$, due to the ubiquitous phenomenon of modulation instability (MI). Numerical evidence is provided in \cite[Appendix C, pg. 903]{blmt2018} for the Hamiltonian non-integrable DNLS ($\gamma=\delta=0$) that the system exhibits the dynamical behavior of modulationally unstable integrable models, like the NLS partial differential equation and the Ablowitz-Ladik (AL) lattice, when supplemented with the boundary conditions \eqref{nv} and the relevant initial conditions of the form of a localized perturbation of a constant background \cite[initial data 8(a)-8(c), pg. 891]{blmt2018}. The potential persistence of this behavior in the presence of gain and loss effects will be a main question which will be investigated and discussed below. On the other hand, due to the physical relevance of the model \eqref{dnls_gl} in the context of optics, it is important to remark that {\em rogue wave-type structures} have been observed in remarkable experiments \cite{Nat2007}, \cite{Nat2016}. In the discrete realm, such structures are described by the famous class of discrete rational solutions of the Ablowitz-Ladik lattice \cite{APS} which satisfy nonzero boundary conditions like \eqref{nv}. The important question of the potential emergence, even robustness of such waveforms in the dynamics of non-integrable DNLS  equations initiated new studies as in \cite{Tsironis2009} in the framework of the Salerno lattice and its generalization which may interpolate between the DNLS and the AL system \cite{SDP2018}; see also \cite {EP}, \cite{WM} for the existence, stability and dynamics of periodic and quasi-periodic solutions on the top of a finite background for this modified Salerno model. The potential emergence of extreme waveforms in the dynamics of the dissipative DNLS \eqref{dnls_gl} will be also investigated in the present work. 

The second case concerns the system supplemented with periodic boundary conditions. For instance, we will consider an arbitrary number of $N+1$ oscillators which are placed equidistantly on the interval $\Omega=[-L,L]$ of length $2L.$ We denote  $k=1/h^2$ the discretization parameter, where $h=2L/N$ is the lattice spacing. The spatial coordinate of the oscillators is given then by $x_n=-L+nh$, $n=0,1,2,\ldots,N$. Then we supplement Eq.~\eqref{dnls_gl}  with periodic boundary conditions 
\begin{equation}\label{eq02}
	u_n =u_{n+N}
\end{equation} 
and initial conditions $u_n(0) = u_{n,0}\in {\ell}^p_{\mathrm{per}}$, the space of periodic sequences (see \cite{GL1} and below). 
We remark that apart of the physical significance of the periodic boundary conditions \eqref{eq02}, the periodic problem serves also as a finite lattice approximation of the infinite lattice problem in the case of the nonzero boundary conditions \eqref{nv} (but also for the vanishing boundary conditions $A=0$)  when $L$ is sufficiently large. 

However, regarding even its local in time solvability, the problem \eqref{dnls_gl}-\eqref{nv} is fundamentally different than its periodic counterpart. We prove that the problem \eqref{dnls_gl}-\eqref{nv} admits a unique, at least a local in time solution {\em if and only if $A=A_*:=\sqrt{-\gamma/\delta}$}. In other words, spatially localized solutions on the top of the finite background $A>0$ exist if and only if $A$ is fixed to the ``critical value'' $A_*$. Apart of its independent interest, this result raises a serious warning concerning the approximations of the infinite lattice problem by the finite lattice, as the periodic one: when $L$ is sufficiently large one may simulate the  dynamics of the problem  with non-zero boundary conditions \eqref{nv} for arbitrary $A>0$, by triggering initial conditions of the form 
\begin{eqnarray}
	\label{apinc}
	\lim_{|n|\rightarrow\infty} u_n(0)=A,
\end{eqnarray} 
with a sufficiently fast rate of decay of  $u_n(0)$ on $A$, 
in order to make the error induced by the periodic boundary conditions negligible since the initial data and the corresponding solutions will satisfy them only asymptotically. It is evident that the numerical solutions which can be produced for an arbitrary $A$, do not capture the crucial restriction for the solvability of the infinite lattice problem stated above. 

In  light of the above result for the infinite lattice, the periodic problem \eqref{dnls_gl}-\eqref{eq02} deserves much attention. It is shown that the relevant finite dimensional dynamical system posseses a global attractor and its spatially averaged $l^2$-norm is uniformly bounded by $A_*$.  
In this paper, to the best of our knowledge, we provide a novel argument that characterizes the global attractor. Specifically, we establish that the global attractor is a time-periodic orbit defined by the unique non-trivial plane wave solution with a constant amplitude denoted as $A_*$, a frequency of $\tilde{\omega}$, and a wave number of $q$. This solution satisfies the dispersion relation given by:
\begin{eqnarray*}
	\tilde{\omega} = 4k\sin^2\Big(\frac{hq}{2}\Big)-A_*^2.
\end{eqnarray*}
To implement this argument, we first demonstrate that the attractor represents the $\pmb{\omega}$-limit set of all initial conditions stemming from solutions in the form of:
\begin{equation*}
	u_w(x_n,t) = A(t)\exp\left[\rmi(q x_n - \Omega (t)\right],\quad t\in\mathbb{R}.
\end{equation*}
What sets our argument apart is that we prove convergence not only in terms of amplitudes, where $A(t)\rightarrow A_*$ as $t\rightarrow\infty$ but also through a detailed analysis of the exact frequency function $\Omega(t)$. We determine its behavior in relation to its slant asymptotes, ultimately establishing that $\Omega(t)\rightarrow \Omega_{\infty}(t) := \tilde{\omega}t$. We extend this result to cover the remaining initial conditions through a finite-dimensional representation using their discrete Fourier series. Secondly, we conduct a MI  analysis for the plane wave attractor, demonstrating its potential for instability. By combining the convergence results with the MI analysis, we establish that this instability is transient.

To illustrate the convergence and transient MI effects, we perform numerical simulations for two representative scenarios: plane-wave initial conditions and localized ones superimposed on a finite background. These numerical findings, which align closely with the theoretical results, also showcase the evolution of the initial conditions' spectrum. Of particular interest is the case of unstable wave numbers, where the system ultimately selects a node from within the stability band upon reaching convergence to the attractor.

The analysis presented above, particularly in the context of numerical simulations for the periodic lattice, underscores that any localized phenomena arising due to  MI are inherently transient. However, within a finite time frame before reaching ultimate equilibrium, these localized phenomena can be notably intriguing.

For localized initial conditions superimposed on top of $A=A_*$, we observe the emergence of waveforms similar to Peregrine solitons as initial events. Subsequently, at later time intervals, spatiotemporal patterns develop due to the MI of the underlying background. These patterns bear a remarkable resemblance to those observed in the dynamics of localized initial data in the integrable Ablowitz-Ladik lattice \cite{blmt2018}. They exhibit the universal structure described in \cite{bm2017} during the nonlinear stage of MI. This structure consists of two outer, quiescent sectors corresponding to the supporting background, separated by a wedge-shaped central region characterized by modulated periodic oscillations. The proximity of these dynamics is noteworthy, considering that DNLS \eqref{dnls_gl} represents a system that has moved ``two steps forward'' in breaking the integrability barrier defined by the AL system. DNLS \eqref{dnls_gl} is a dissipative perturbation of the non-integrable Hamiltonian DNLS. The stability of the quiescent sectors can be fully justified by the analytical results. Notably, $A=A_*$ represents the amplitude of the supporting background, which, as highlighted earlier, is the only amplitude supporting localized solutions in the infinite lattice. Moreover, for the periodic lattice used in the numerical simulations, $A=A_*$ corresponds to the amplitude of the globally asymptotically stable background. In the case where $A\neq A_*$ where solutions for the infinite lattice do not exist, the numerical simulations for the periodic lattice show that the MI pattern persists on the top, this time of the evolving background.  This continues until the amplitude reaches its limiting value $A_*$, at which point the entire MI structure disappears and is replaced by the attracting plane wave.

Motivated by the above results and our current investigations on the closeness between integrable and non-integrable lattice systems \cite{DJNa, DJNc,DJNb}, we examine the potential proximity between the solutions of the Ablowitz-Ladik lattice and the DNLS system \eqref{dnls_gl} when both systems are supplemented with periodic boundary conditions. While for bounded solutions the distance is always uniformly bounded, in the same fashion as in \cite{DJNa, DJNc,DJNb}, we prove that for finite times this distance grows at most linearly. Since the first events reminiscent to rogue waves occur at small time intervals, this growth is proved to be effective in measuring the distance between the rogue wave events emerging in the dynamics of \eqref{dnls_gl} and those defined by the analytical  discrete Peregrine soliton solution of the Ablowitz-Ladik lattice.

The paper is structured as follows. In Section~\ref{NVBC}, we provide the proof of the existence of solutions for the infinite lattice with nonzero boundary conditions. For the definition and properties of the sequence spaces that define the functional analytic set-up of the problem, we refer to \cite{GL1} and the references therein. In Section~\ref{GlobalA}, we present the results concerning the dynamics of the periodic lattice. Sections~\ref{num1} and \ref{num2} are devoted to the numerical results. In particular, in Section~\ref{num1}, we present the results of the numerical studies concerning the transient MI of the plane wave attractor. In Section~\ref{num2}, we present the numerical results for localized initial conditions, distinguishing between the case $A=A_*$, the only case permitted for the infinite lattice, and the case $A\neq A_*$, which can be relevant in the periodic lattice. In this section, we also provide the proof of the estimates for the distance between the solutions of the DNLS \eqref{dnls_gl} and the Ablowitz-Ladik periodic lattices, along with the corresponding numerical investigations. We conclude with Section~\ref{conclusions}, which summarizes our findings.

\section{The case of the infinite lattice with nonzero boundary conditions}
\label{NVBC}
Solutions of the initial value problem for the infinite lattice \eqref{dnls_gl} with the nonzero boundary conditions  and $\gamma>0$, $\delta<0$ exist if and only if $A=A_*$. This is the main result of this section stated in the following theorem. 
\begin{theorem}
\label{TH1}
Consider the DNLS system \eqref{dnls_gl} with $\gamma>0,\delta<0$ in the infinite lattice supplemented with the nonzero boundary conditions \eqref{nv}. We also assume that the initial condition satisfies \eqref{apinc}. Then the system has a solution if and only if $A=A_*$. The solution is unique. 
\end{theorem}
\begin{proof}
First, we apply to the DNLS \eqref{dnls_gl} the gauge transformation
\begin{equation}
	\label{gauge1}
u_n(t)=\psi_n(t)\exp(\rmi A^2t),	
\end{equation}
(in order to make the  boundary conditions time independent), yielding the system in the form
\begin{equation}
	\label{dnls_gla}
\rmi\dot{\psi}_n  + k(\psi_{n+1} - 2\psi_n+\psi_{n-1})-A^2\psi_n + |\psi_n|^2 \psi_n = \rmi\gamma \psi_n + \rmi\delta |\psi_n|^2 \psi_n,
\end{equation}
which satisfies the boundary conditions
\begin{equation}
	\label{Vbc}
	\lim_{|n|\rightarrow\infty} \psi_n(t) =A.
\end{equation}
Next, we apply to the system \eqref{dnls_gla}, the change of variables
\begin{equation}
	U_n={\psi}_n-A, \,\,\, n\in {\mathbb{Z}},\label{shift}
\end{equation}
The system for the new variable $U_n$ reads as
\begin{align}
\label{dnls_glb}
\mathrm{i}\dot{U}_n+k(U_{n+1} - 2U_n+U_{n-1})-A^2(U_n+A) &+|U_n+A|^2(U_n+A)\nonumber\\ &=\mathrm{i}\gamma (U_n+A)+\mathrm{i}\delta|U_n+A|^2(U_n+A),\;\;
\end{align}
where $U_n$, due to the non-zero boundary conditions \eqref{Vbc}, should satisfy {\em the vanishing boundary conditions at infinity}
\begin{equation}
	\label{Vbc1}
	\lim_{|n|\rightarrow\infty} U_n(t) =0.
\end{equation}
Thus, the system should be examined for its solvability in the sequence spaces $\ell^p$,  for any $p\geq 1$. Let us recall the crucial embedding 
\begin{equation}
	\label{lp1}
	\ell^q\subset\ell^p,\quad \|u\|_{\ell^p}\leq \|u\|_{\ell^q}, \quad 1\leq q\leq p\leq\infty.
\end{equation}	
For the local existence of solutions, we may apply the generalized Picard-Lindel\"{o}f Theorem \cite[Theorem 3.A, pg. 78]{zei85a}. 
The discrete Laplacian $$\Delta_dU_n=k(U_{n+1} - 2U_n+U_{n-1}),$$ is a bounded linear operator $\Delta_d:\ell^p\rightarrow \ell^p$, for any $p\geq 1$, that is, there exists a constant $C>0$, such that
\begin{equation}
	\label{genlp8}
	\|\Delta_d U\|_{\ell^p}\leq C\|U\|_{\ell^p},\quad \mbox{for all}\quad U\in\ell^p.
\end{equation}
Next, we consider  the nonlinear operators
\begin{align*}
\mathcal{F}_1(U_n)&=-A^2(U_n+A)+|U_n+A|^2(U_n+A),\\
F_2(U_n)&=\mathrm{i}\gamma (U_n+A)+\mathrm{i}\delta|U_n+A|^2(U_n+A).
\end{align*}
For the operator $\mathcal{F}_1$, we observe that 
\begin{align*}
\mathcal{F}_1(U_n)&=-A^2U_n-A^3+(|U_n|^2+AU_n+A\overline{U}_n+A^2)(U_n+A)\\
&=|U_n|^2 U_n+A(U_n^2+2|U_n|^2)+A^2(U_n+\overline{U}_n).
\end{align*}
Using the inequality 
\begin{equation*}
	\left| |u_n|^2u_n-|v_n|^2v_n\right|\leq |v_n|^2|u_n-v_n|+|u_n|\left(|v_n|+|u_n|\right)|u_n-v_n|,
\end{equation*}	
and the embedding \eqref{lp1}, we may show that $\mathcal{F}_1:\ell^p\rightarrow\ell^p$, for any $p\geq 1$, is well defined and is Lipschitz continuous on bounded sets of $\ell^p$. That is, for any $U,V\in B_R$, with $B_R$ a closed ball of $\ell^p$ of center $0$ and radius $R$, we derive the existence of constants $K_1(A, R)$ and $K_{1,1}(A, R)$, such that 
\begin{align}
	\label{op1}
	\| \mathcal{F}_1(U)\|_{\ell^p}&\leq	K_{1}(A, R),\\
	\label{op2}
	\| \mathcal{F}_1(U)-\mathcal{F}_1(V)\|_{\ell^p}&\leq	K_{1,1}(A, R)\|U-V\|_{\ell^p}.
\end{align}
Similarly, for the operator $\mathcal{F}_2$, we have that 
\begin{equation}
\label{devq0}	
	\mathcal{F}_2(U_n)=\mathrm{i}\gamma U_n+\mathrm{i}\gamma A+\mathrm{i}\delta A^3+\mathrm{i}\delta\left\{|U_n|^2 U_n+A(U_n^2+2|U_n|^2)+A^2(2U_n+\overline{U}_n\right\}.
\end{equation}
Assume that $A=A_*$. Then, $\mathrm{i}\gamma A+\mathrm{i}\delta A^3=0$ and we can show that the operator $\mathcal{F}_2$ has the properties \eqref{op1}-\eqref{op2}. We may proceed  by recasting the system \eqref{dnls_glb} in the form of the abstract evolution equation 
\begin{align}
	\label{semifl1}
	\dot{U}=\mathcal{L}(U),\quad \mathcal{L}(U)=\mathrm{i}[-\Delta_d(U)-\mathcal{F}_1(U)+\mathcal{F}_2(U)],
\end{align}
and we will implement the aforementioned  Picard-Lindel\"{o}f Theorem on its equivalent integral formula
\begin{equation}
	\label{mildsl1}
	U(t)=U(0)+\int_{0}^{t}\mathcal{L}[U(s)]ds,
\end{equation}
to establish the existence of a unique solution $U\in l^p$, for any $p\geq 1$, when the initial condition $U(0)\in \ell^p$. This is the case, since the operator $\mathcal{L}:\ell^p\rightarrow\ell^p$ has the properties \eqref{op1}-\eqref{op2} and the conditions of the Picard-Lindel\"{o}f Theorem are satisfied. The solution has the following properties:  There exists some $T^*(U(0))>0$ such that the  corresponding initial value problem for \eqref{dnls_glb}-\eqref{Vbc1}, has a unique solution which is continuously differentiable with respect to time, i.e., $U\in C^1([0,T],\ell^p)$  for all $0<T<T^*(U(0))$. In addition, the following alternatives hold: Either $T^*(U(0))=\infty$ (global existence) or $T^*(U(0))<\infty$ and $\lim_{t\uparrow T^*(U(0))}\|U(t)\|_{\ell^p}=\infty$ (collapse). Furthermore the solution $U$ depends continuously on the initial condition $U(0)\in \ell^p$, with respect to the norm of $C([0,T],\ell^p)$. Using the change of variables \eqref{gauge1} and \eqref{shift}, we establish the existence of a unique solution of the original problem \eqref{dnls_gl}-\eqref{nv},
\begin{equation}
\label{sol}
u_n(t)=(U_n(t)+A_*)\exp(\rmi A_*^2t).
\end{equation}
For the reverse direction of the proof, we have to show that if a solution of the initial boundary value \eqref{dnls_gl}-\eqref{nv}-\eqref{apinc} problem exists, then $A= A_*$. We argue by contradiction: Assume that a solution of the problem \eqref{dnls_gl}-\eqref{nv}-\eqref{apinc} exists for all $t\in [0,T]$, for some $T>0$ and that  $A\neq A_*$. Then, through the same transformations \eqref{gauge1}-\eqref{shift}, a solution $U$ of the problem \eqref{dnls_glb}-\eqref{Vbc} should exist for all $t\in [0,T]$,  and  $U\in C^1([0,T],\ell^p)$, which means that $\|\dot{U}(t)\|_{\ell^p}<\infty$, for all $t\in [0,T]$. Hence, due to \eqref{semifl1}, $\|\mathcal{L}(U)\|_{\ell^p}<\infty$ for all $t\in [0,T]$. This can't be valid because if  $A\neq A_*$, the operator $\mathcal{F}_2$ is not bounded on bounded sets of $\ell^p$ as it follows from \eqref{devq0}. Indeed, for $U$ the solution of the problem which satisfies $\|U(t)\|_{\ell^p}\leq R$, for all $t\in [0,T]$, the boundedness of $\mathcal{F}_2$ in bounded sets of $\ell^p$ would imply the counterpart of  \eqref{op1},
\begin{equation}
\label{op3}	
\| \mathcal{F}_2(U(t))\|_{\ell^p}\leq	K_{2}(A, R),\quad \mbox{for all}\quad  t\in [0,T].
\end{equation}
Since from \eqref{devq0},
\begin{equation}	
\label{op4}
\mathrm{i}\gamma A+\mathrm{i}\delta A^3=-\mathcal{F}_2(U_n)+\mathrm{i}\gamma U_n+\mathrm{i}\delta\left\{|U_n|^2 U_n+A(U_n^2+2|U_n|^2)+A^2(2U_n+\overline{U}_n\right\},
\end{equation}
\eqref{op3} and the fact that the rest of the terms of the right-hand side of \eqref{op4} are bounded in $\ell^p$ for all $t\in [0,T]$, would imply that the left hand side of \eqref{op4} is summable, and this  is impossible. Therefore, the  initial-boundary value problem  \eqref{dnls_gl}-\eqref{nv}-\eqref{apinc} has a solution if and only if $A=A_*$.
\end{proof}
\subsection{Remarks on Theorem \ref{TH1} for finite lattice approximations}
The finite lattice approximations are defined either in the spaces of periodic sequences with  period $N$,  denoted by
\begin{equation*}
	{\ell}^p_{per}:=\bigg\{U=(U_n)_{n\in\mathbb{Z}}\in\mathbb{R}:\quad U_n=U_{n+N},\quad
	\|U\|_{\ell^p_{\mathrm{per}}}:=\bigg(h\sum_{n=0}^{N-1}|U_n|^p\bigg)^{\frac{1}{p}}<\infty\bigg\}, \quad 1\leq p\leq\infty,
\end{equation*}	
or in the case of the Dirichlet boundary conditions, in the finite dimensional subspaces of $\ell^2$,
\begin{equation*}
	{\ell}^p_{0}:=\bigg\{U=(U_n)_{n\in\mathbb{Z}}\in\mathbb{R}:\quad U_0=U_{N}=0,\quad
	\|U\|_{\ell^p_{0}}:=\bigg(h\sum_{n=1}^{N-1}|U_n|^p\bigg)^{\frac{1}{p}}<\infty\bigg\}, \quad 1\leq p\leq\infty.
\end{equation*}
If $\mathcal{X}^p$ is either one of the above spaces, their norms  are  equivalent, according to the inequality
\begin{equation}
	\label{equi}
	\|U\|_{\mathcal{X}^q}\leq \|U\|_{\mathcal{X}^p}\leq N^{\frac{(q-p)}{qp}}\|U\|_{\mathcal{X}^q},\quad 1\leq p\leq q<\infty.
\end{equation}
In this finite dimensional set-up, it is evident from the definition of the operators $ \mathcal{F}_1$ and $\mathcal{F}_2$ that Theorem \ref{TH1}  is valid for any $A$ since the operators $\mathcal{F}_1,\mathcal{F}_2:\mathcal{X}^p\rightarrow\mathcal{X}^p$ are well defined and locally Lipschitz continuous for any $A$. Therefore, finite lattice approximations do not capture the essential necessary and sufficient conditions on  $A$,  for the existence of solutions of the infinite lattice which is proved in Theorem \ref{TH1}.
\section{The case of the lattice with periodic boundary conditions}
\setcounter{equation}{0}
\label{GlobalA}
In this section we consider the case of the finite lattice \eqref{dnls_gl} when supplemented with the periodic boundary conditions \eqref{eq02}. This finite dimensional system, apart of its physical significance may serve as one of the main  finite lattice approximations of the problem of the infinite lattice supplemented with the nonzero boundary conditions discussed in the previous section. In this connection, the value of the background amplitude $A_*$, which is the only one which may support localized solutions on top of the background in the infinite lattice, chiefly governs the dynamics. The phase space of the system is
the Hilbert space $\ell^2_{per}$ endowed with the scalar product 
\begin{equation*}
	\label{lp2}
	(\phi,\psi)_{\ell^2_{\mathrm{per}}}=h\,\mathrm{Re}\sum_{n=0}^{N-1}\phi_n\overline{\psi}_n,\quad\phi,\,\psi\in\ell^2_{\mathrm{per}}.
\end{equation*}
 
\begin{theorem}\label{T1}~\\ \vspace{-0.5cm}
	\begin{enumerate}
		\item
(\textit{The global attractor in $\ell^2_{per}$}). Let $u_n(0)=u_{n,0} \in\ell^2_{per},$ and  $\gamma>0$, $\delta<0.$ Then, the solution of the initial value problem for \eqref{dnls_gl}--\eqref{eq02} $u_n\in C^1([0,\infty),\ell^2_{per})$. Furthermore, 
for the dynamical system 
\begin{equation}
\label{wds1}
\varphi(t, u_{n,0}): \ell^2_{per}\rightarrow \ell^2_{per},~~\varphi(t, u_{n,0})=u_n(t),
\end{equation}
there exists a global attractor and the averaged power $\di P_a[u(t)]:=\frac{1}{N}\sum_{n=0}^{N-1}|u_n(t)|^2$  of the solutions satisfy the estimate
\begin{equation}\label{bound1}
	\lim\sup_{t\to\infty}P_a[u(t)]\leq A_*^2.
\end{equation}
\item  The system \eqref{dnls_gl}-\eqref{eq02} admits  solutions of the form  
\begin{equation}\label{plane}
	w_n(t):=w(x_n,t) = A(t)\exp\left[\rmi(q x_n - \Omega (t)\right],\quad t\in\mathbb{R},
\end{equation}
where   $q = K\pi/L$, $K\in\bbN$ is the constant wavenumber and the time variable frequency function $\Omega(t)$ and  amplitude $A(t)$ satisfy the system of equations
\begin{align}\label{dispersion}
	\Omega(t) & = 4k\sin^2\Big(\frac{hq}{2}\Big)t-\Theta (t), \quad \dot{\Theta}(t)=A^2(t)\\
	\dot{A}(t) & = \gamma A(t) + \delta A^3(t).
	\label{eqbern}
\end{align}
for all $t\in\mathbb{R}$.  The solutions \eqref{plane} have the following property under the flow \eqref{wds1}: The limit set
\begin{equation}
	\label{limitp}
	\pmb\omega(w_n(0))=\mathcal{C}_*,
\end{equation} 
where $\mathcal{C}_*$ is the time-periodic orbit defined by the only non-trivial plane wave solution of constant amplitude $A_*$ with frequency $\tilde{\omega}$ and wave number $q$ satisfying \eqref{dispersion}, when we set in \eqref{dispersion}  $A(t)\equiv A_*$, that is,
\begin{equation}
\label{pwfor15th}
	\tilde{\omega} = 4k\sin^2\Big(\frac{hq}{2}\Big)-A_*^2.
\end{equation}
For the solutions \eqref{plane}, the limiting bound \eqref{bound1} is exact, in the sense that
\begin{equation}\label{boundp}
	\lim_{t\to\infty}P_a[w_n(t)]= A_*^2.
\end{equation}
\end{enumerate}
\end{theorem}
\begin{proof}
1. We multiply equation \eqref{dnls_gl} by $-\mathrm{i}u_n$ in the $\ell^2_{per}$-inner product. Using that $\di (\Delta_{d}u,\mathrm{i}u)_{\ell^2_{\mathrm{per}}}=-\mathrm{i}\sum_{n=0}^{N-1}|u_{n+1}-u_{n}|^2$,  by keeping the real  parts of the resulting equation, we derive the following power balance law
\begin{equation}\label{balance_eq}
\frac{d}{dt}\bigg(h\sum_{n=0}^{N-1}|u_n|^2\bigg) = 2\gamma\, h\sum_{n=0}^{N-1}|u_n|^2 +2\delta\, h\sum_{n=0}^{N-1}|u_n|^4.
\end{equation}
Note that for $\gamma = \delta=0$ we have the conservation of the power, which is one of the  conserved quantities of the conservative DNLS . 

Let $\gamma>0$ and $\delta:=-\tilde{\delta}<0$.   The second term of the right-hand side of Eq.~\eqref{balance_eq} can be estimated as 
$$
-2\tilde{\delta}\,h\sum_{n=0}^{N-1}|u_n|^4\leq -\frac{2\tilde{\delta}}{Nh}\bigg(h\sum_{n=0}^{N-1}|u_n|^2\bigg)^2,
$$ 
where we used the Cauchy--Schwarz inequality
$$
\sum_{n=0}^{N-1}|u_n|^2 \le \sqrt{N}\bigg(\sum_{n=0}^{N-1}|u_n|^4\bigg)^{1/2}.
$$
 Consequently,  Eq.~\eqref{balance_eq} turns to a Bernoulli's differential inequality 
\begin{equation}\label{eq23}
\frac{d}{dt}\bigg(h\sum_{n=0}^{N-1}|u_n|^2\bigg)\leq 2\gamma\,h\sum_{n=0}^{N-1}|u_n|^2-\frac{2\tilde{\delta}}{Nh}\bigg(h\sum_{n=0}^{N-1}|u_n|^2\bigg)^2,
\end{equation}
and equivalently, when diving by $Nh$, the equation for 
$P_a[u(t)]$
\begin{equation}\label{eq24}
\frac{d}{dt}P_a[u]\leq 2\gamma P_a[u]-2\tilde{\delta}P^2_a[u].
\end{equation}
Setting $\phi= 1/P_a[u]$ we see that $\phi$ satisfies 
$$
\frac{d}{dt}\phi+2\gamma \phi\geq 2\tilde{\delta},
$$
which (with the aid of the integrating factor $\exp(2\gamma t)$) gives
$$
\phi(t)\geq \phi(0)\exp(-2\gamma t)+\frac{\tilde{\delta}}{\gamma}\left[1-\exp\left(-2\gamma t\right)\right]$$
and finally 
\begin{equation}\label{eq27}
P_a[u(t)] \leq \frac{1}{P_a[u(0)]^{-1}\exp(-2\gamma t)+\frac{\tilde{\delta}}{\gamma}\left[1-\exp\left(-2\gamma t\right)\right]}, \quad  \forall t\geq 0.
\end{equation}
The latter means that  $P_a[u(t)]$ is uniformly bounded. 
Letting $t\to\infty,$ we conclude that  $P_a[u(t)]$ satisfies the estimate \eqref{bound1}.

From \eqref{bound1}, we deduce that the dynamical system  \eqref{wds1} has bounded orbits $\forall t\in[0,\infty)$. Let $\mathcal{B}:=\{u_{n,0} \in\ell^2_{per}:P_a[u] \leq R^2\}$ be an arbitrary closed ball in $\ell^2_{per}$.  
The ball $\mathcal{B}_a(0,\rho):=\{u\in\ell^2_{per}: P_a[u]\leq\rho^2,\;\;\;\rho^2>A_*^2\}$ is an absorbing set for the dynamical system  $\varphi(t, u_{n,0})$: there exists $T^*(\mathcal{B},\mathcal{B}_a)>0$ such that  for any $t\geq T^*(\mathcal{B},\mathcal{B}_a)$ it holds that $P_a[u]\leq \rho^2$ and $\varphi(t, \mathcal{B})\subset\mathcal{B}_a.$
Hence, we may define the $\pmb{\omega}$-limit set in $\ell^2_{per}$, for any bounded set $\mathcal{B}$, 
\begin{equation*}
\pmb{\omega}(\mathcal{B}) = \bigcap_{s\geq 0}\overline{\bigcup_{t\geq s}\varphi(t, \mathcal{B})}.
\end{equation*}
Since the dynamical system is finite-dimensional the above limit set defines its compact global attractor. 

2. We start by seeking solutions of the form
\begin{equation}
	\label{pwfor1}
u_n(t)=W(t)\exp\left[\rmi(q x_n - \tilde{\omega} t)\right],\quad \tilde{\omega}\in\mathbb{R}.
\end{equation}
By substituting \eqref{pwfor1} to \eqref{dnls_gl} we derive the equation for $W(t)$
\begin{equation}
	\label{pwfor2}
\rmi\dot{W}+\lambda W+|W|^2W=\rmi\gamma W+\rmi\delta|W|^2W,\quad \lambda=\tilde{\omega}-4k\sin^2\Big(\frac{hq}{2}\Big).
\end{equation}
The linear term $\lambda W$ will be removed by the gauge transformation
\begin{equation}
	\label{pwfor3}
W(t)=\Phi(t)\exp(\rmi\lambda t),
\end{equation}
applied to \eqref{pwfor2}. Then, the equation for $\Phi(t)$ is
\begin{equation}
\label{pwfor4}
\rmi\dot{\Phi}+|\Phi|\Phi=\rmi\gamma \Phi+\rmi\delta |\Phi|^2\Phi.
\end{equation}
The last transformation we will apply is 
\begin{equation}
	\label{pwfor5}
\Phi(t)=A(t)\exp[\rmi\Theta (t)],\quad A,\Theta\in C^1(\mathbb{R}),
\end{equation}
to \eqref{pwfor4}.  We derive the equation for $A(t)$,
\begin{equation}
	\label{pwfor6}
\rmi\dot{A}-\dot{\Theta}A+A^3=\rmi\gamma A+\rmi\delta A^3.
\end{equation}
The choice
\begin{equation}
	\label{pwfor7}
	\dot{\Theta}=A^2,
\end{equation}
eliminates the second and the third term of the left-hand side of \eqref{pwfor6},  in order to obtain \eqref{eqbern}. Then returning to the starting form of solutions \eqref{pwfor1}, by applying the transformations \eqref{pwfor3} and \eqref{pwfor5} one after another, we build the solution \eqref{plane} with the system \eqref{dispersion}-\eqref{eqbern}. Note that in this process, the term $\tilde{\omega} t$ was eliminated resulting in equation \eqref{dispersion} for $\Omega (t)$. We could start with the ansatz $u_n(t)=W(t)\exp(\rmi q x_n)$ instead of \eqref{pwfor1} to build exactly the same solution, however we started with \eqref{pwfor1} for notational purposes relevant to the solution of constant amplitude $A(t)=A_*$ which will be discussed below.  Indeed, we will discuss the limit of the orbit \eqref{plane} as $t\rightarrow\infty$, by using \eqref{dispersion}-\eqref{eqbern}. 

Let $A(0)$ be the initial amplitude of the solution \eqref{plane}. 
Equation \eqref{eqbern} has the unique solution 
\begin{equation}\label{eqf}
	A^2(t) = \frac{\gamma A^{2}(0)}{(\gamma  +\delta A^{2}(0)) \exp(-2\gamma t) - \delta A^{2}(0)},
\end{equation} 
for which, it holds that
\begin{equation}\label{limit}
	\lim_{t\to\infty}A^2(t)=A_*^2.
\end{equation}
Next, by integrating \eqref{pwfor7} in the interval $[0,t]$, we get 
\begin{equation}
\label{pwfor9}
\Theta (t)-\Theta (0)=\int_{0}^tA^2(s)ds.
\end{equation}
The function $\Theta (t)-\Theta (0)$ has the following properties:
\begin{align}
	\label{pwfor10}
	\lim_{t\rightarrow\infty}\big(\Theta (t)-\Theta (0)\big)&=\infty,\\
		\label{pwfor11}
	\lim_{t\rightarrow+\infty}\frac{\Theta (t)-\Theta (0)}{t}&=\lim_{t\rightarrow\infty}\dot{\Theta}(t)=A_*^2,\\
		\label{pwfor12}
	\lim_{t\rightarrow\infty}\big[\big(\Theta (t)-\Theta (0)\big)-A^2_*t\big]&=-\frac{1}{2\delta}\ln\Big(\frac{ A^2(0)}{A_*^2}\Big)=: b>0.
\end{align}
 Note that when $A^2(0)=A_*^2$, we have $A^2(t)=A_*^2$ for all $t\geq 0$, and in this case,  we have that $b =0$.  From \eqref{pwfor11} and \eqref{pwfor12}, we have that as $t\rightarrow \infty$, the function $\Theta (t)$ has the slant asymptote 
 \begin{equation}
 	\label{pwfor13}
 \Theta_{\infty}(t)=A_*^2t+\Theta(0)+b.
 \end{equation} 
Without loss of generality we can select $\Theta(0)=-b$ (since we can always adjust by the relevant phase factor). Then, in the limit $t\rightarrow\infty$, we replace the function $\Omega (t)$ by 
\begin{align}
\Omega_{\infty}(t)&=4k\sin^2\Big(\frac{hq}{2}\Big)t
-\Theta_{\infty} (t)=4k\sin^2\Big(\frac{hq}{2}\Big)t-A_*^2t=\tilde{\omega} t,
\label{pwfor14}\\
\tilde{\omega}&= 4k\sin^2\Big(\frac{hq}{2}\Big)-A_*^2. \label{pwfor15}
\end{align}  
Now, we may use \eqref{limit} in the solution \eqref{plane} and \eqref{pwfor14}-\eqref{pwfor15} (or alternatively \eqref{pwfor6} and \eqref{pwfor11}), to see that as $t\rightarrow \infty$, $\mathcal{C}_*$ is the orbit defined by the solution
\begin{equation}
	\label{pwfor16}
	w_{\infty,n}(t)=A_*\exp\left[\rmi(q x_n - \Omega_{\infty} (t)\right]=A_*\exp\left[\rmi(q x_n - \tilde{\omega} t\right], 
\end{equation}
where $\tilde{\omega}$ and $q$ satisfy  \eqref{pwfor15}. This is exactly the claimed dispersion relation \eqref{pwfor15th} for the unique plane wave of constant amplitude $A_*$ \eqref{pwfor16}, to exist. Usually we derive the dispersion relation \eqref{pwfor15} by substitution of the solution \eqref{pwfor16} to the DNLS \eqref{dnls_gl}.  Here, it is interesting to find, that the solution \eqref{pwfor16} can be derived when using \eqref{eqf} for $A(0)^2=A_*^2$ for its amplitude, and importantly, by the interpretation of the function $\Theta_{\infty}(t)$ for $\beta=0$ and $\Theta (0)=0$ in \eqref{pwfor13}. The assertion \eqref{limitp} is proved. 
The proof of \eqref{boundp} follows by the calculation of 
\begin{equation}
	\label{equi2}
P_a[w_n(t)]=\frac{1}{N}\sum_{n=0}^{N-1}|w_n(t)|^2=A^2(t),
\end{equation}
and the limit \eqref{limit}.

\end{proof}
\noindent{\bf Remark:} It is important to stress the following: The system \eqref{dnls_gl} do not admits solutions of the form \eqref{plane} with constant amplitude $A$, unless this amplitude is $A=A_*$. {\em In other words, there is no other constant $A>0$ than $A_*$, which satisfies both equations \eqref{dispersion} and \eqref{eqbern} simultaneously}.   On the other hand, we may always trigger initial conditions of the form 
\begin{equation}
\label{pwin}
u_n(0)=A\exp(\rmi qx_n), \quad A\neq A_*,
\end{equation}
and arbitrary $q$, which can be considered generically as perturbations of the unique  plane wave solution of constant amplitude.
The fate of such initial conditions will be determined in the next paragraph.   
\subsection{Consequences of Theorem \ref{T1} with respect to a linear instability analysis}
\label{TransA}
The main consequence of Theorem \ref{TH1} is that the dynamical system defined by the periodic lattice \eqref{dnls_gl}-\eqref{eq02} is {\em globally asymptotically stable} for any initial condition in $\ell^2_{per}$: for arbitrary initial conditions all the orbits are converging to a state whose averaged norm satisfies the bound \eqref{bound1}. In the special case where the initial condition is defined by a solution of the form \eqref{plane}, the attractor is trivial, defined by the only plane wave solution which can have constant amplitude, which is $A^*$. This global asymptotic stability do not excludes the emergence of instability effects. Nevertheless, these effects {\em can be only transient}, prior the convergence of the system to its ultimate state. We  investigate the potential of these transient instability effects in the framework of the linear MI analysis of plane waves. The motivation  is that the fate of such initial conditions is predicted by the Theorem \ref{TH1} as stated above.  

The modulation instability of  plane-waves is well studied in the case of the Hamiltonian DNLS \cite{Kivshar}. We will perform here this analysis for the system \eqref{dnls_gl} with the presence of the linear gain and nonlinear loss effects.  The significance of these findings lies on the fact that they prove that {\em  the simplest initial data of the form  \eqref{pwin} have interesting dynamics prior the convergence to the attractor, and that even this attractor can be transiently, modulationally unstable}.

To justify the above statement, we consider the perturbation of the only plane wave solution \eqref{plane} with constant amplitude $A(t)\equiv A_*$, 
\begin{equation}
\label{p1}
u_n(t) = (A_*+b_n(t))\exp\left[\rmi (\theta_n(t)+\psi_n(t))\right],
\end{equation}
where $\theta_n(t)=q x_n-\tilde{\omega} t$, $x_n = -L + nh, n=0,1,\ldots,N$ and $b_n,\psi_n$ are the amplitude and the angle perturbations, respectively.  Substitution of \eqref{p1} to \eqref{dnls_gl}, yields the system 
\begin{equation}
	\label{p2}
	\left\{
	\begin{aligned}
		- A_*\dot{\psi}_n&-kA_*(\psi_{n+1}-\psi_{n-1})\sin(hq)+k( \Delta_db_n )\cos(hq)+2A_*^2b_n=0,\\
		\dot{b}_n&+kA_*(\Delta_d\psi_{n})\cos(hq)+k( b_{n+1}-b_{n-1})\sin(hq) -2\delta A_*^2b_n=0.
	\end{aligned}
	\right.
\end{equation} 
For the derivation of \eqref{p2}, we used  the dispersion relation \eqref{pwfor15} and that $\gamma+\delta A_*^2=0$.

To simplify the system \eqref{p2}, we  consider amplitude perturbations of the form $b_n(t) = b_0\exp(\rmi(Q x_n - \Omega_p t))$ and phase perturbations of the form $\psi_n(t) = \psi_0\exp(\rmi(Q x_n - \Omega_p t))$. Then the system \eqref{p2} becomes 
\begin{equation}
	\label{p4}
\left\{
\begin{aligned}
	&\left(\Omega  -2k \sin(hQ)\sin(hq)\right) A_*\psi_0\rmi +(-4 k\sin^2(hQ/2)\cos(hq)+2A_*^2)b_0=0,\\
	& -4kA_*(\sin^2(hQ/2)\cos(hq)\psi_0\rmi + \left(\Omega-2k\sin(hQ) \sin(hq)  -  2\delta A_*^2\rmi\right)b_0 = 0.
\end{aligned}
\right.
\end{equation}
The linear homogeneous system \eqref{p4} for $(b_0,\psi_0)$  has a nontivial solution if and only if its determinant is zero. This request provides the  modulation instability formula: 
\begin{equation}\label{MI}
	\Lambda^2 -  2\rmi\delta A_*^2 \Lambda -\Gamma\left(\Gamma-2A_*^2 \right) = 0,
\end{equation}
with 
$\Lambda=\Omega-2k \sin(hQ)\sin(hq) $ and  $\Gamma=4k\sin^2(Q/2)\cos(hq)$. 
We elaborate further on the formula \eqref{MI}, to derive the conditions for the MI of the examined plane wave solution. 
Since \eqref{MI} is a quadratic equation in $\Lambda$ with complex coefficients, it has only complex roots.
The discriminant of \eqref{MI} is
$$
\Delta_{\Lambda} = 4(\Gamma(\Gamma - 2 A_*^2 ) - \delta^2 A_*^4),
$$ 
and accordingly, the roots are 
\begin{equation}
	\Lambda_{\pm} = \delta A_*^2 \rmi \pm \frac{\sqrt{\Delta_{\Lambda}}}{2} = \delta A_*^2 \rmi \pm  \sqrt{\Gamma(\Gamma - 2 A_*^2 ) - \delta^2 A_*^4}.
\end{equation}
We are interested in the case $\text{Im}(\Omega) = \text{Im}(\Lambda_{\pm}) > 0$ for the exponential growth of the perturbations, and thus the emergence of instability.  If the quantity under the square root is positive, we have that 
$$
\text{Im}(\Lambda_{\pm})=\delta<0.
$$
If the quantity under the square root is negative, we have that 
$$
\text{Im}(\Lambda_{-})=\delta -\sqrt{\delta^2 A_*^4- \Gamma(\Gamma - 2 A_*^2 )}<0.
$$

Thus, we are only interested for the case
\begin{align}
	\text{Im}(\Lambda_{+})> 0 \iff &  \delta A_*^2 +  \sqrt{ \delta^2 A_*^4 - \Gamma(\Gamma - 2 A_*^2 )}  > 0 \nonumber\\ 
	\label{negIm}	
	\iff &  \Gamma(\Gamma - 2 A_*^2 )<0. \\
	\iff & 4k\sin^2(hQ/2)\cos(hq)[4k\sin^2(hQ/2)\cos(hq) - 2A_*^2] <0.\nonumber
\end{align}
If $\cos(hq) < 0$, then \eqref{negIm} is violated and the plane wave solution is stable. 
Therefore, instability occurs when 
\begin{equation}
	\label{conq}
\cos(hq) > 0.
\end{equation}
Indeed, in order to verify \eqref{conq}, it suffices to show that there exists at least one $Q$ such that \eqref{negIm} holds. This $Q=Q^*$ satisfies the equation 
$$
\sin(hQ^*/2) = \frac{h^2A_*^2}{2\cos(hq)},
$$
and implies the requested condition for instability \eqref{negIm}, if \eqref{conq} holds.

We will illustrate the above transient MI of the plane wave attractor with numerical simulations, in the next section. 
\section{Numerical study I: Transient modulation instability effects and convergence of plane wave initial conditions to the plane wave attractor}
\setcounter{equation}{0}
\label{num1}
This section is devoted to the first part of the numerical investigations of the paper. We illustrate the transient MI of the plane wave attractor, which is verified by the \textcolor{red}{analysis} of the section \ref{TransA}. For this purpose, we study the dynamics of the simplest plane wave initial conditions, defining perturbations of the plane wave attractor,  

\begin{equation}\label{plane_ic}
	u_n(0) = (A_*+A_p)\exp\Big(\frac{\rmi K\pi x_n}{L}\Big),
\end{equation} 
where $A_p$ is the perturbation of the amplitude $A_*$ and $K$ is the wave number of the plane wave initial  condition \eqref{plane_ic}. The spatial parameters of the lattice are fixed to $L=50$,  $h=1$, thus $N=100$. The linear gain and nonlinear loss strengths are $\gamma=-\delta=1.5$. These values of $\gamma$ and $\delta$ are chosen to achieve a faster convergence to the attractor.

Due to the periodic boundary conditions \eqref{eq02} and discreteness, the physically meaningful $K$ are integers in the interval $[0,N/2]$. In terms of the analysis of section \ref{TransA}, $q=\pi K/L$, and the MI condition \eqref{conq} is satisfied for  $0<K<25$, defining the unstable modes. The rest of the modes $K$  are modulationally stable.
\begin{figure}[h!]
	\centering 
	\begin{tabular}{ccc}
		(a)&(b)&\hspace{0.3cm}(c)\\
		\includegraphics[height=4cm]{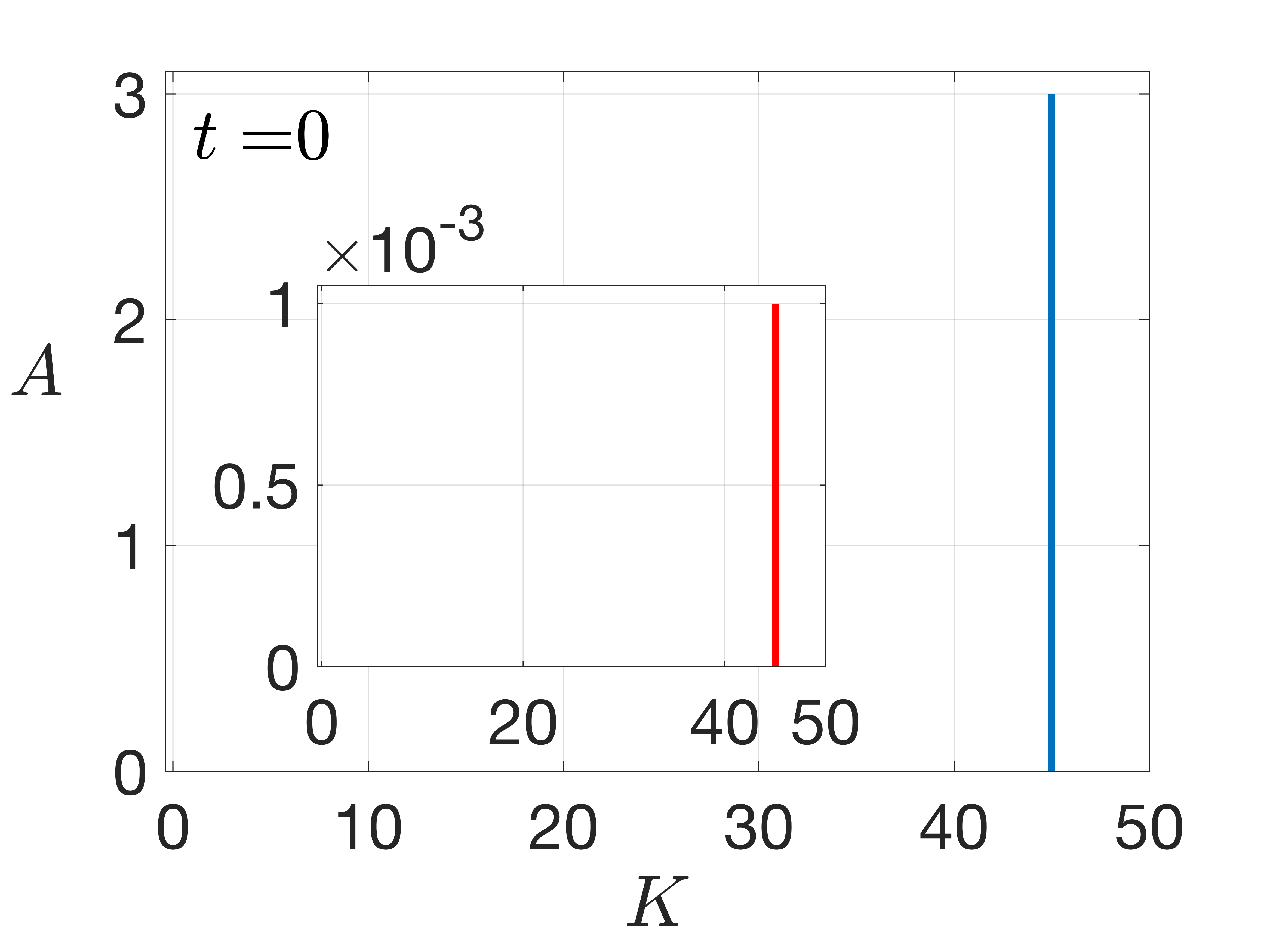}& 
		\includegraphics[height=4cm]{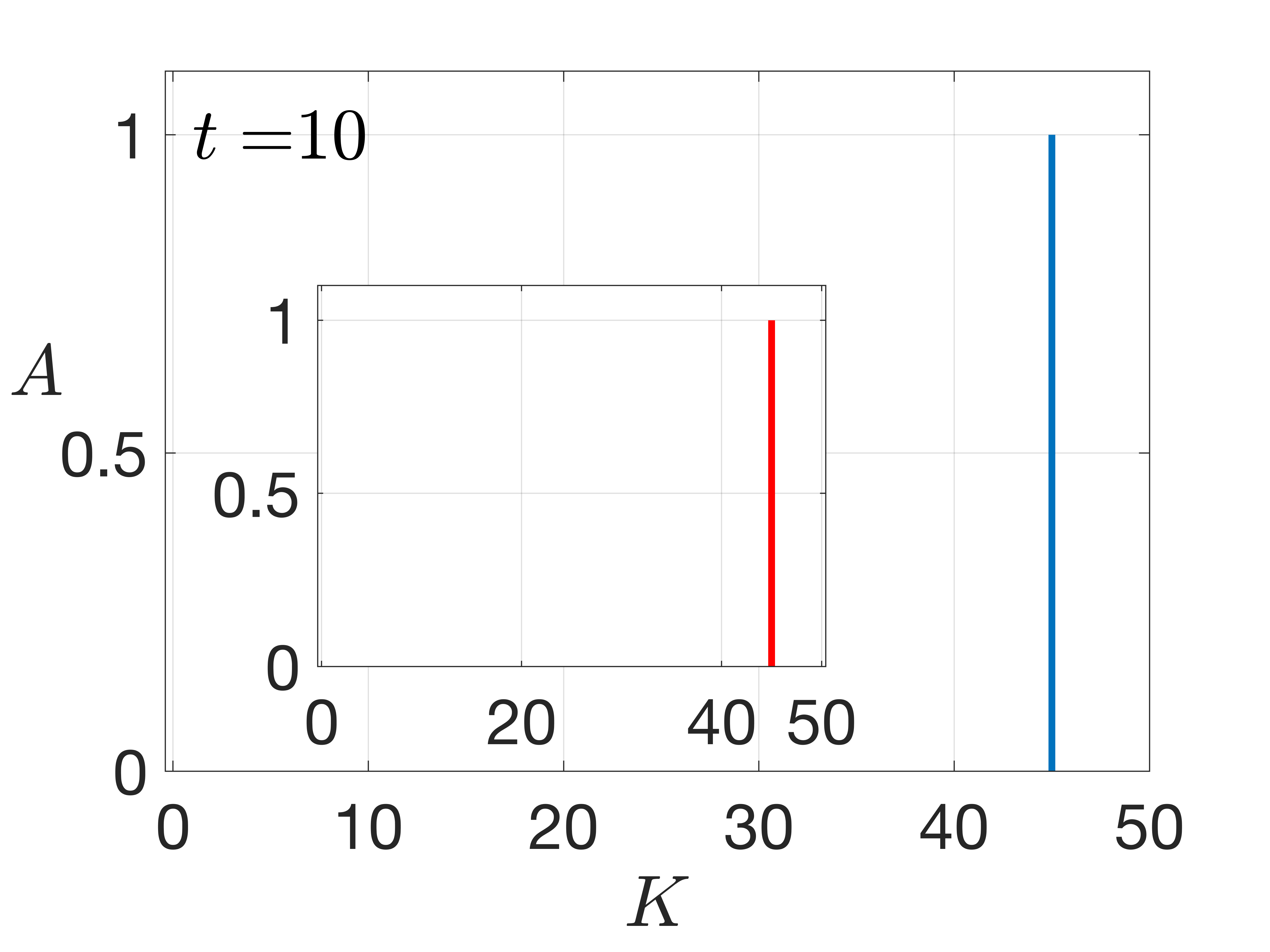}&		\hspace{0.3cm}\includegraphics[height=4cm]{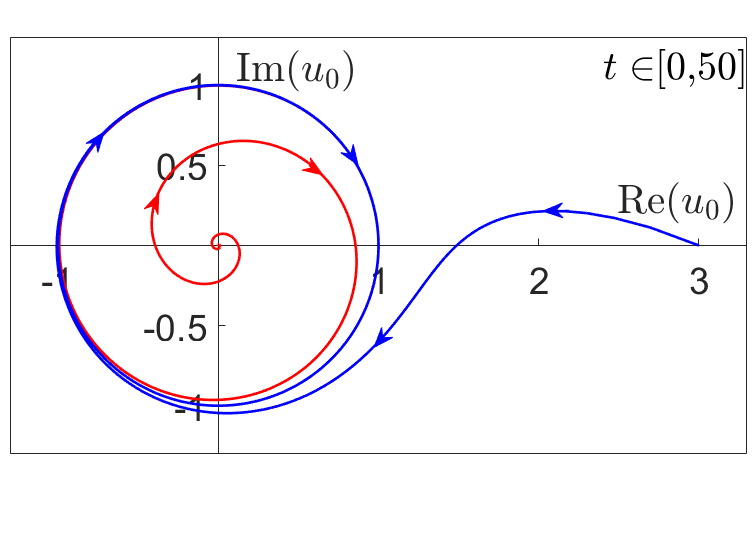}
	\end{tabular}
	\caption{Convergence to the plane wave attractor with amplitude $A_*=1$ for the stable mode $K=45$. Rest of parameters: $L=50$, $h=1$, $\gamma=-\delta=1.5$.  Panel (a): Fourier spectrum of the initial condition \eqref{plane_ic} for $A_p=2$ and $A_p=-0.999$.  Panel (b): Fourier spectrum of the solutions at $t=10$, for $A_p=2$ and $A_p=-0.999$. Panel (c): Convergence of the orbits of the central node of the lattice  $u_0(t)$ in the phase plane $\big(\textrm{Re}(u_0(t)),\,\textrm{Im}(u_0(t))\big)$ to the limit cycle of radius $A_*=1$; internal (red) orbit for $A_p=-0.999$ and external (blue) orbit for $A_p=2$.  Details are given in the text.}
	\label{Fig5}
\end{figure}

The first example concerns the dynamics of the  \textit{stable modes}. We choose the stable mode $K=45$. According to the results of Section~\ref{GlobalA}, for arbitrary perturbations $A_p$, the orbit starting from the initial condition \eqref{plane_ic} will converge to the plane wave attractor with amplitude $A_*=1$ and the initially chosen mode $K=45$.  The dynamics for the amplitude perturbations $A_p=2$ and $A_p=-0.999$ is depicted in Fig.~\ref{Fig5}.  Panel (a) depicts the Fourier spectrum of the initial condition \eqref{plane_ic} for $A_p=2$ and $A_p=-0.999$, where $K=45$ is the chosen stable mode. Panel (b) depicts the Fourier spectrum of the solutions at $t=10$ where convergence to the attractor is attained, for both cases of the perturbations $A_p$. It shows that the attractor preserves the initial wave number $K=45$, as predicted. No other modes are excited transiently (intermediate snapshots are omitted as they are identical to (b)), since the initial mode is stable.  Panel (c) illustrates the convergence to the attractor in the plane $\big(\mathrm{Re}(u_0(t)), \mathrm{Im}(u_0(t))\big) $  which is the phase plane for the central node $u_0(t)$ of the lattice  located at $n=0$. The dynamics is portrayed as a convergence to the asymptotically stable limit cycle of radius $A_*$; the orbit in the interior of the limit cycle is the one for $A_p=-0.999$ and the orbit in the exterior of the limit cycle is the one for $A_p=2$.   

\begin{figure}[tbh!]
	\centering 
	\begin{tabular}{cc}
		(a)&\hspace{0.7cm}(b)\\	
		\includegraphics[scale=0.25]{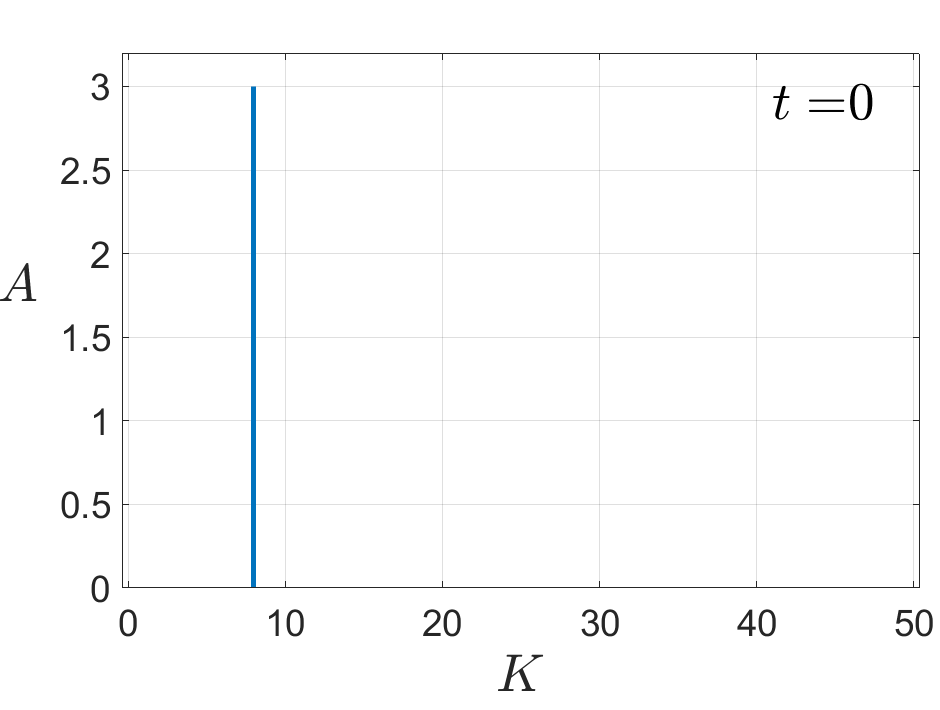}&
		\hspace{0.3cm}\includegraphics[scale=0.25]{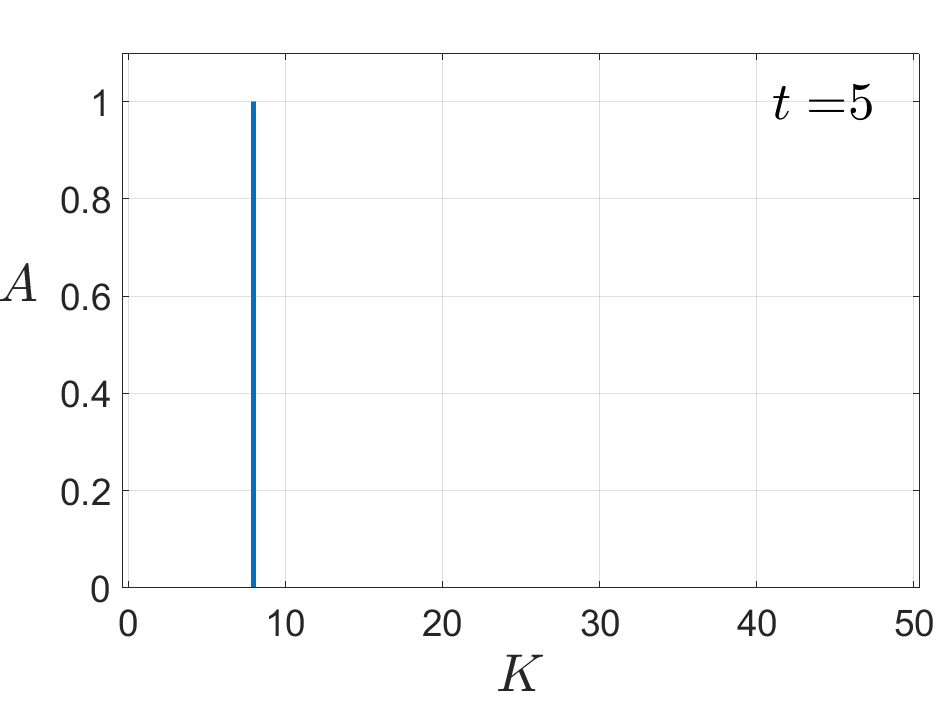}\\[5pt]
		(c)&\hspace{0.7cm}(d)\\
		\includegraphics[scale=0.25]{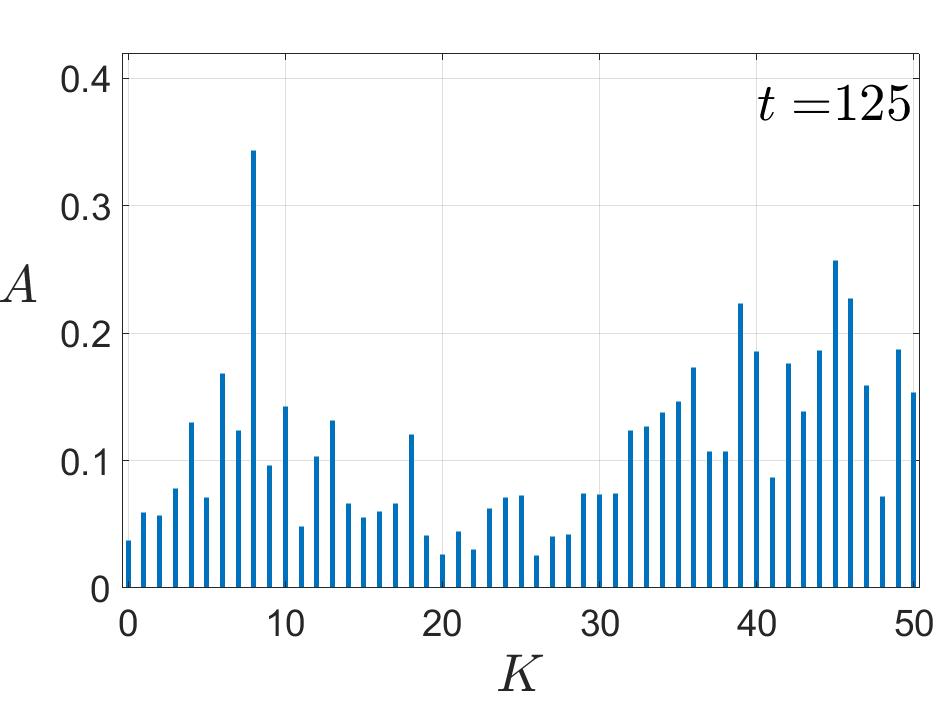}&
		\hspace{0.3cm}\includegraphics[scale=0.25]{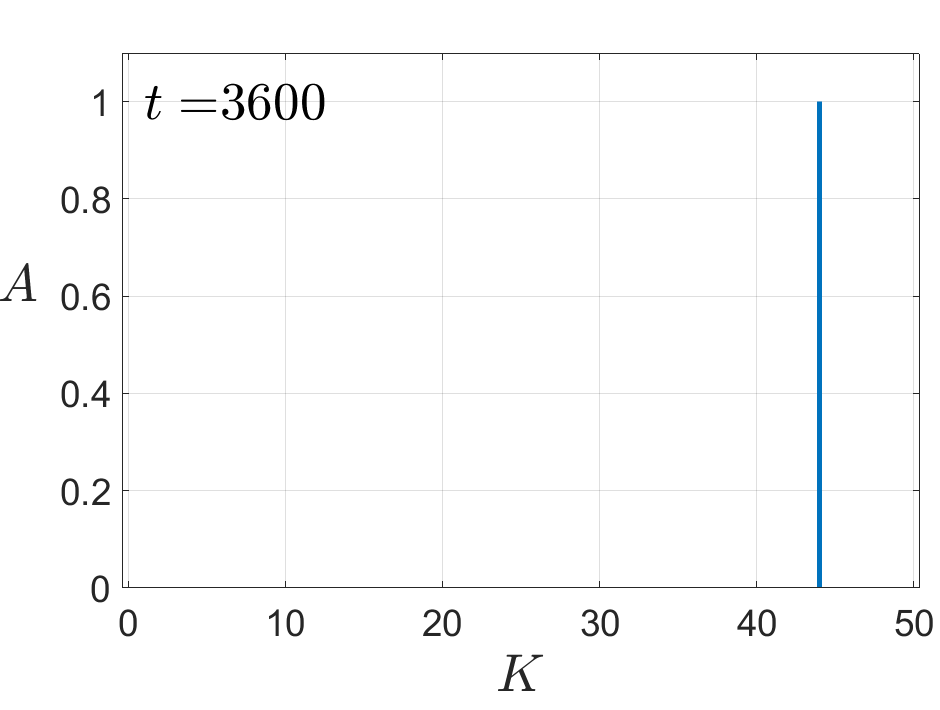}
	\end{tabular}
	\caption{ Snapshots of the evolution of the Fourier spectrum of the solution with initial condition \eqref{plane_ic} for $A_p=2$ and $K=8$. Rest of parameters: $L=50$, $h=1$, $\gamma=-\delta=1.5$. Details are given in the text.}
	\label{Fig6}
\end{figure}
The second example concerns the case of the  {\em unstable modes}.  We choose the unstable mode $K=8$.  According to the results of Section~\ref{TransA} the corresponding orbit should converge to the plane wave attractor with amplitude $A_*=1$. Concerning the wave number of the attracting state, the analytical arguments provide the prediction for its choice:  For the linear gain strength $\gamma=1.5$ the convergence to the amplitude $A_*$ should be fast. On the other hand, the initial unstable wave number should introduce transient MI effects. While the convergence to the steady-state amplitude $A_*$ has a fast exponential rate,  the emergence of instability effects will occur later, perturbing the phase of the attracting state, transiently. Nevertheless, the ultimate state should choose a stable wave number, since it is the globally asymptotically stable state for the system.  This prediction is fully illustrated in Figures \ref{Fig6} and \ref{Fig6b} depicting the dynamics for $A_p=2$. The dynamics for the corresponding unstable case of $A_p=-0.999$ is very similar and, thus, is not shown herein. Figure \ref{Fig6} shows the evolution of the Fourier spectrum of the solution. The initial wave number $K=8$, of  panel (a) is preserved as shown in panel (b) for $t=5$; at this time the system has almost converged to the plane wave with amplitude $A_*=1$ at a fast exponential rate as shown in the panel (a) of Figure \ref{Fig6b}, and MI effects are not yet visible. These are developed later as depicted in panel (c) of Figure \ref{Fig6}, demonstrating the excitation of the whole Fourier spectrum at $t=125$. Panel (b) of Figure \ref{Fig6b} portrays the deformed orbit  for  $t\in[0,125]$  which is a time period within the stage of MI.  The instability can last for large times, however, eventually,  the global attractor is attained and its wave number should be selected by the stability band. This is exactly the case, depicted in  panel (d) of Figure \ref{Fig6}, showing the snapshot of the Fourier spectrum at $t\in[3600,3700]$. At this time the only active wavenumber is the stable mode $K=44$ which is the one selected by the plane wave attractor of amplitude $A_*$. 

\begin{figure}[t!]
	\centering 
	\begin{tabular}{ccc}
		(a)&\hspace{0.5cm}(b)&\hspace{1cm}(c)\\
		\includegraphics[height=3.4cm]{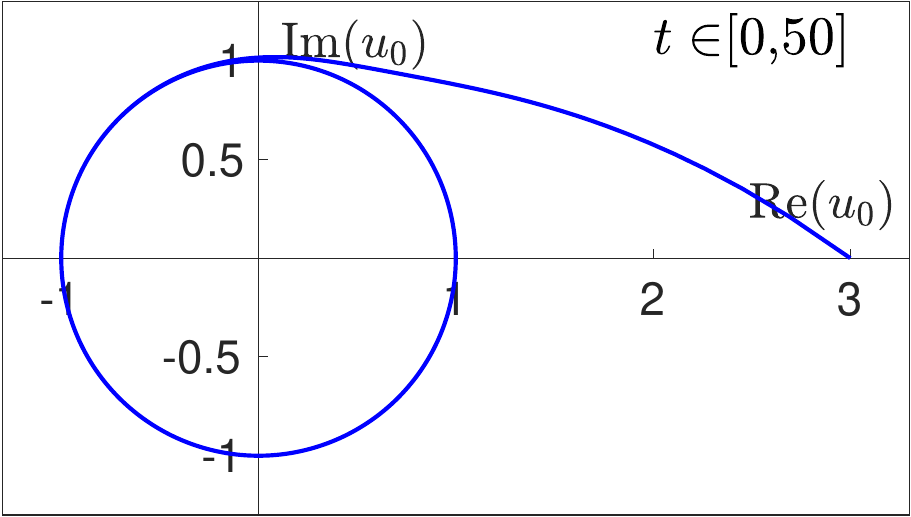}
		&
		\hspace{0.5cm}\includegraphics[height=3.4cm]{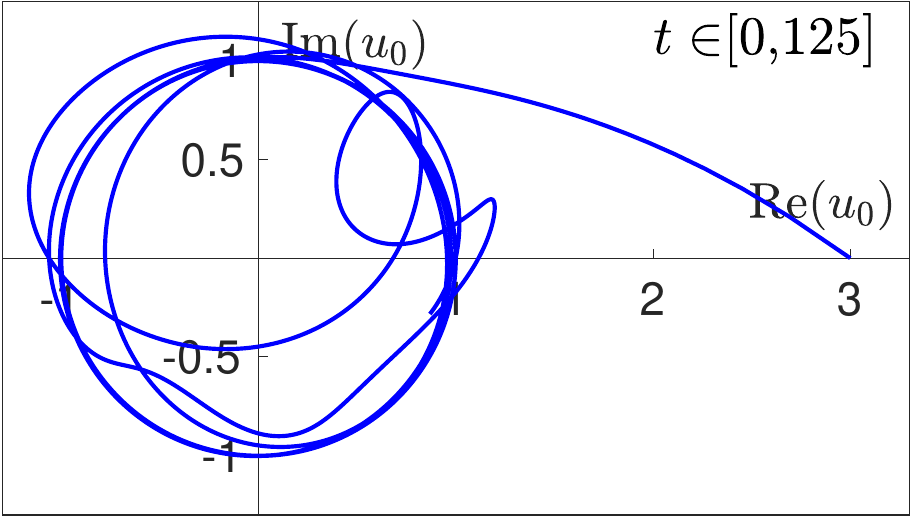}&
		\hspace{0.5cm}\includegraphics[height=3.4cm]{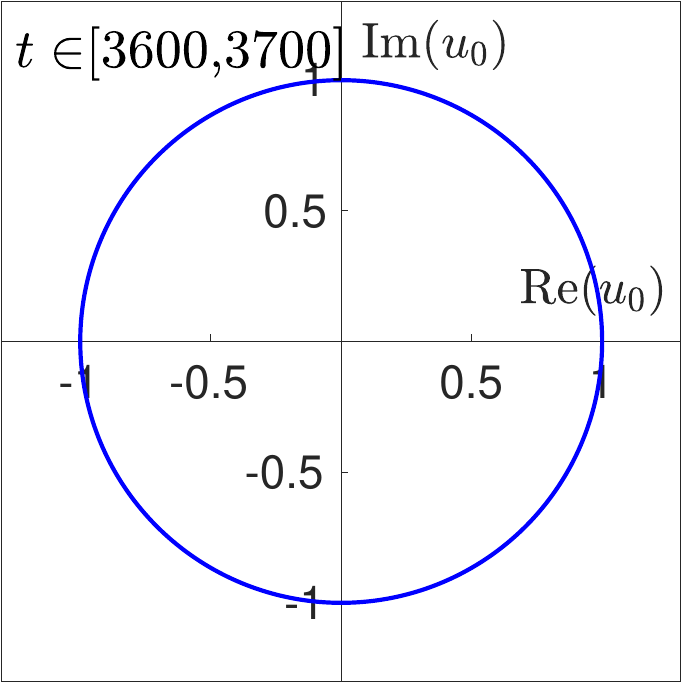}
	\end{tabular}
	\caption{Snapshots of the orbits of the central node of the lattice  $u_0(t)$ in the phase plane $\big(\textrm{Re}(u_0(t)),\,\textrm{Im}(u_0(t))\big)$ for the solution of Figure \ref{Fig6}. Details are given in the text.}
	\label{Fig6b}
\end{figure}
\section{Numerical Study II: Simulating the infinite lattice with nonzero boundary conditions by the finite lattice with periodic boundary conditions }
\label{num2}
This section is devoted to the second part of the numerical investigations of the paper. We study the dynamics of localized initial conditions on the top of a background $A>0$ of the form
\begin{equation}\label{loc_ic}
	u_n(0) = A+f(x),
\end{equation} 
where $f(x)$ is a decaying function $\lim_{|x|\rightarrow\infty}f(x)=0$. The case of localized data is relevant to the simulations of the infinite lattice \eqref{dnls_gl}-\eqref{nv} and thus, they will be explored in detail regarding their global and transient dynamics.
\subsection{Localized initial data \texorpdfstring{\eqref{loc_ic}}{} and the global attractor}
\label{subsec1}
The analytical results of sections \ref{GlobalA} and \ref{TransA} corroborated with the numerical findings of section \ref{num1}, provide a strong indication that the global attractor for all the initial data under the dynamical system \eqref{wds1} is the plane wave attractor \eqref{limitp}. 
These indications are enhanced by the fact that any solution of the DNLS \eqref{dnls_gl} in $\ell^2_{per}$ can be represented by the finite sum of plane waves which is its discrete Fourier series,
 \begin{align*}
 	\label{ser1}
u_n(t)&=\frac{1}{N}\sum_{K=0}^{N-1}A_K(t)\exp\Big(\frac{\rmi K 2\pi n}{N}\Big),\\
A_K(t)&=h\sum_{n=0}^{K-1}u_n(t)\exp\Big(-\frac{\rmi K 2\pi n}{N}\Big)
\end{align*}
and the fact that any plane wave initial condition is attracted by the orbit  $\mathcal{C}_*$.

Figure \ref{Fig8} depicts the dynamics of the Fourier spectrum of the solution for the initial condition \eqref{loc_ic} with
\begin{equation}
	\label{alg1}
	f(x)=\frac{\lambda_1}{\lambda_2+\lambda_3 x^2},\quad \lambda_i>0,\quad i=1,2,3,
\end{equation}
the standard example of an initial condition decaying on $A$ at a quadratic rate. For the initial condition \eqref{loc_ic}-\eqref{alg1}, we choose $A=0.5$, $\lambda_1=\lambda_2=1,\lambda_3=4$. The gain and loss parameters are $\gamma=-\delta=0.1$; note that $A\neq A_*$. The rest of the lattice parameters are $L=50$, $h=1$, $N=100$. 

\begin{figure}[tbh!]
	\centering 
	\begin{tabular}{ccc}
(a)&(b)&(c)\\		
	\includegraphics[scale=0.22]{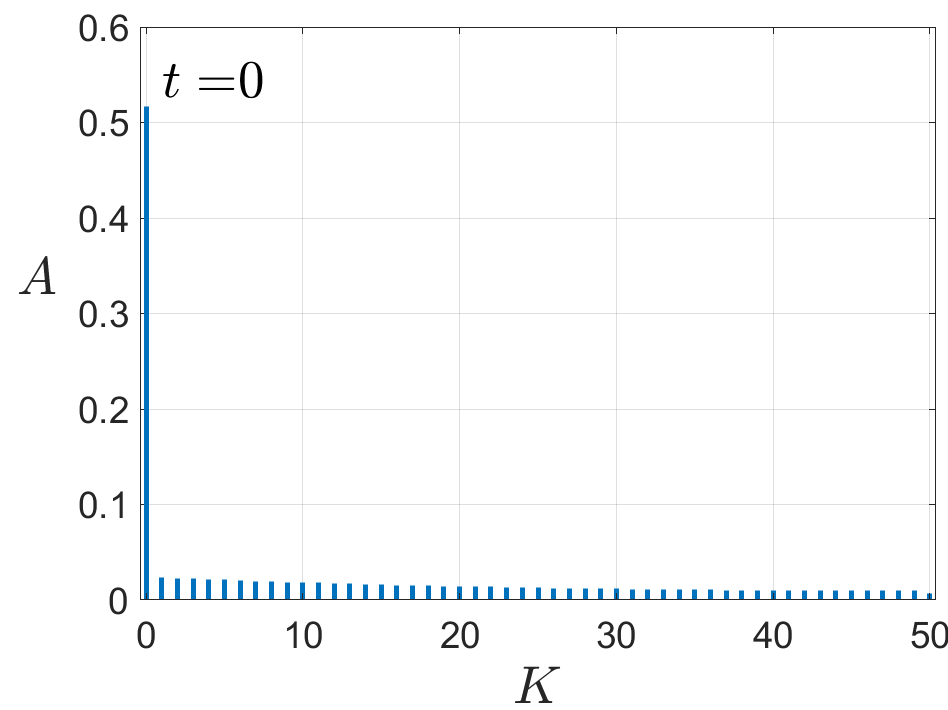}
	&\includegraphics[scale=0.22]{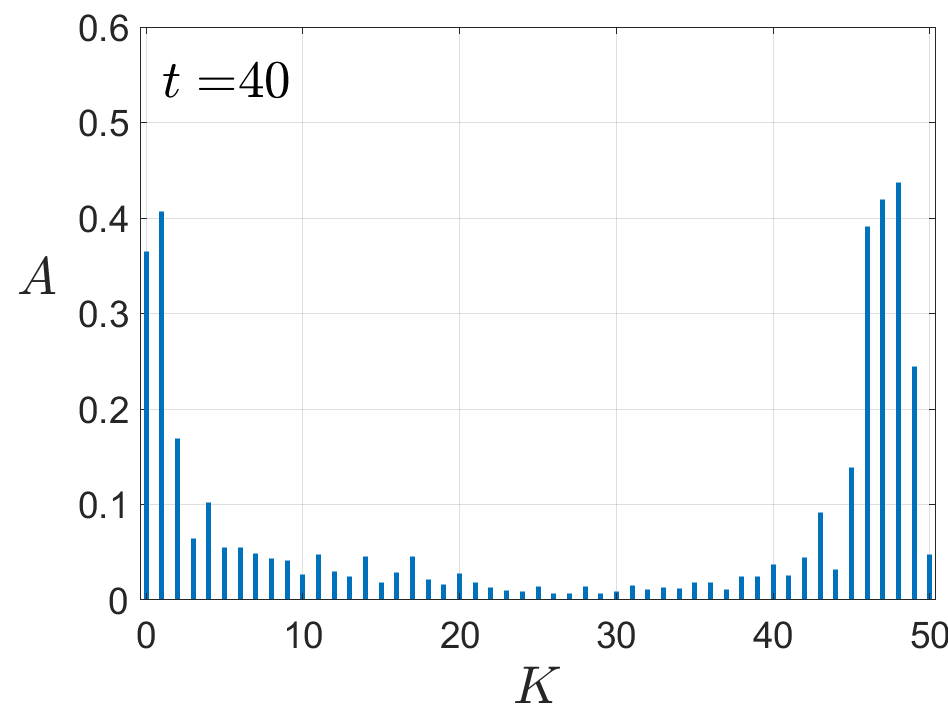}&\includegraphics[scale=0.22]{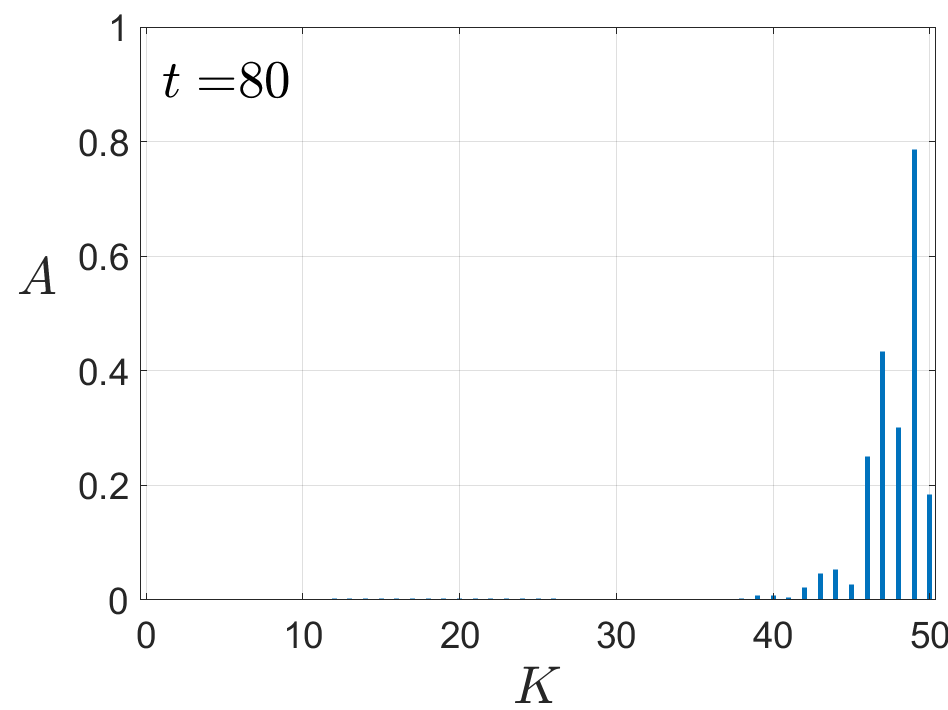}
	\\[1ex]
	(d)&(e)&(f)\\
	\includegraphics[scale=0.22]{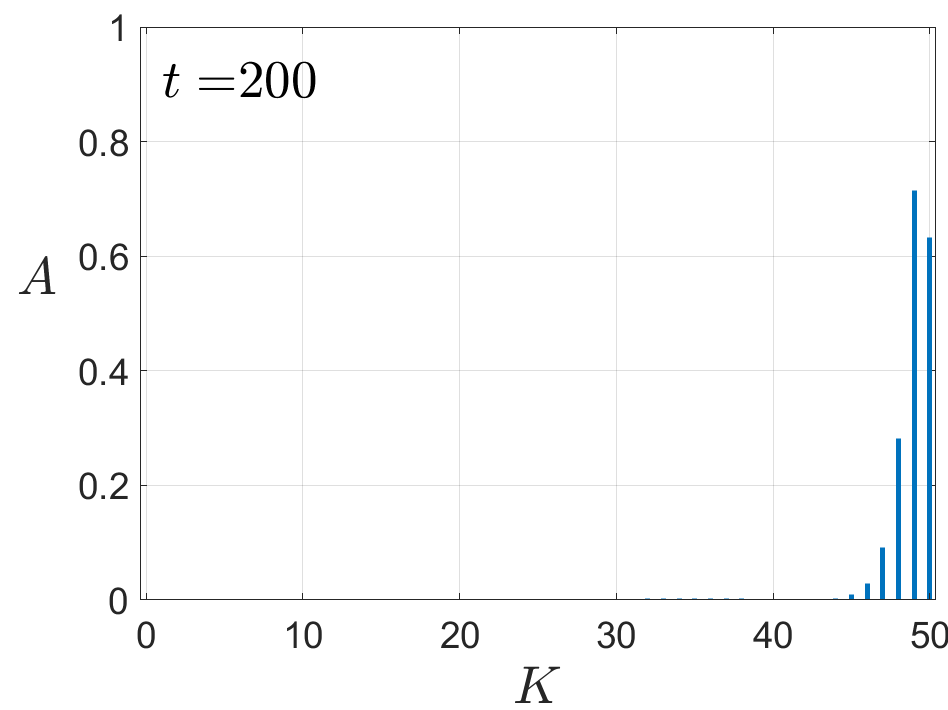}
	&\includegraphics[scale=0.22]{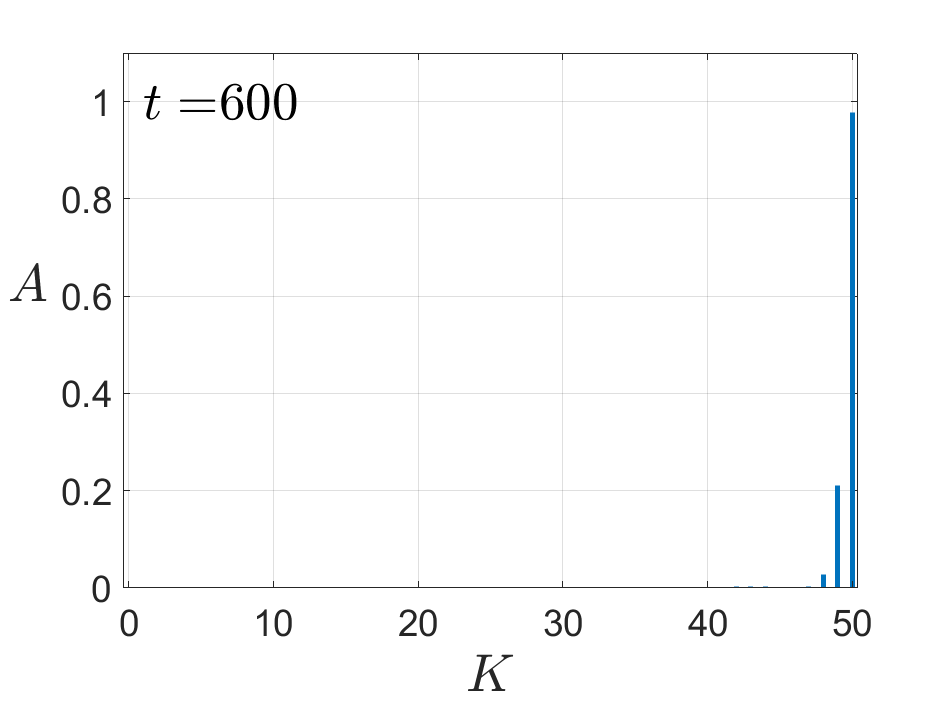}
	&\includegraphics[scale=0.22]{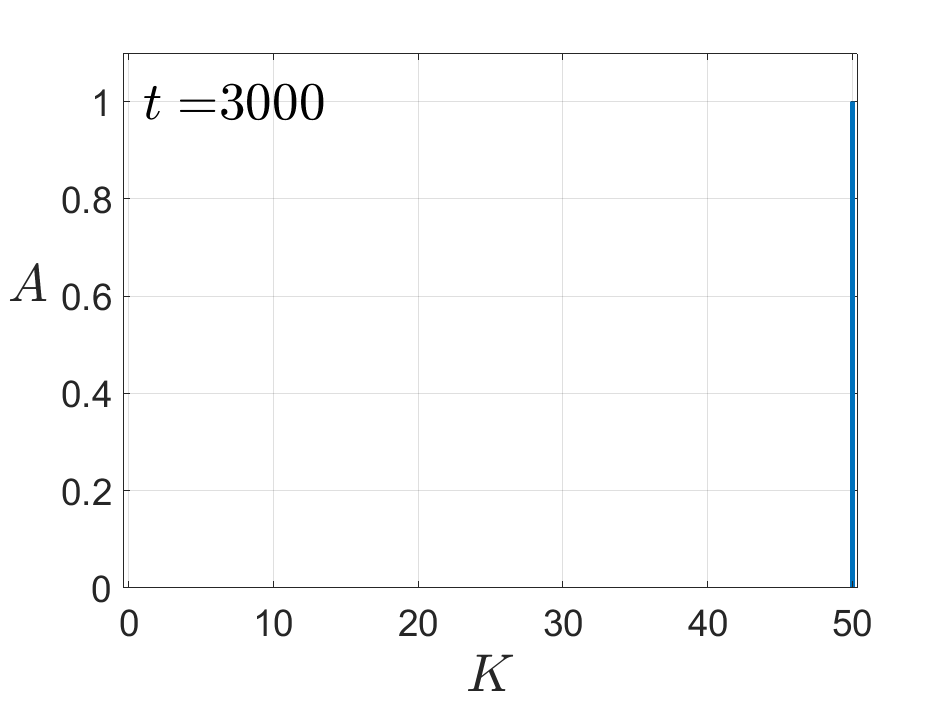}
	\end{tabular}
	\caption{Snapshots of the Fourier spectrum of the solution for the initial condition  \eqref{loc_ic}-\eqref{alg1} with $A=0.5$.  Lattice parameters:, $L=50$, $h=1$, $\gamma=-\delta=0.1$. Details  are given in the text.}
	\label{Fig8}
\end{figure}

The transient dynamics prior to the convergence to $\mathcal{C}_*$ is illustrated. This involves the excitation of the full spectrum at the early stages of the evolution as it is shown in  panel (b) for $t=40$, while progressively, the active Fourier modes are reduced, as it is shown in the snapshots (c) for $t = 80$, (d) for $t=200$ and (e) for $t=600$.  When the orbit $\mathcal{C}_*$ is reached, the only active mode is the stable mode $K=50$, see panel (f). Qualitatively, the same dynamics was observed for other choices of the lattice parameters and gain and loss strengths. 

We can summarize the results regarding the convergence of all initial data to the plane wave attractor as follows: Theorem \ref{T1} and the MI analysis of section \ref{TransA} theoretically predict that the amplitude and spectrum of the initial condition will undergo a transient evolution before converging to the attractor. Ultimately, this process selects the amplitude $A_*$ and a wavenumber within the stable modes of the spectrum. It is important to note that this prediction is qualitative rather than quantitative, as the theoretical results do not specify the exact stable mode of the attractor.
\subsection{Transient dynamics of localized initial data \texorpdfstring{\eqref{loc_ic}}{}: MI dynamics and excitation of extreme wave events in the dynamics of the DNLS \texorpdfstring{\eqref{dnls_gl}}{}.}
\label{RWE}
For the periodic lattice \eqref{dnls_gl}-\eqref{eq02} the numerical simulations illustrate that the dynamics prior the convergence to the global attractor are associated with the emergence of characteristic patterns of MI dynamics and the emergence of extreme wave events. This is a consequence of the transient MI discussed above. Our numerical investigations consider two examples of initial conditions of the form \eqref{loc_ic}, with different rates of localization. The first concerns  \eqref{loc_ic} with the quadratically decaying function $f(x)$ \eqref{alg1}, and the second with the exponentially decaying function
\begin{equation}
	\label{expini}
	f(x)=\sigma\mathrm{sech}(\rho x),\quad \sigma,\;\rho>0.
\end{equation} 
In the light of Theorem \ref{TH1} we distinguish between two cases for $A$ in \eqref{loc_ic}. The case where $A=A_*$ and the case where $A\neq A_*$.
\paragraph{The case $A=A_*$.} We discuss the dynamics for the lattice parameters $h=1$, $N=400$, $L=200$ and gain and loss strengths  $\gamma=0.0025$, $\delta=-0.01$.  According to Theorem \ref{TH1}, the infinite lattice supplemented with the non-zero boundary conditions \eqref{nv} support localized solutions if and only if $ A=A_*=\sqrt{-\gamma/\delta}=0.5$. This is the suitable choice of the background of the initial condition \eqref{loc_ic}, in order to simulate the dynamics of the infinite lattice.
\begin{figure}[tbp!]
	\centering 
	\begin{tabular}{cc}
		(a)&\hspace{0.7cm}(b)\\	
		\includegraphics[scale=0.4]{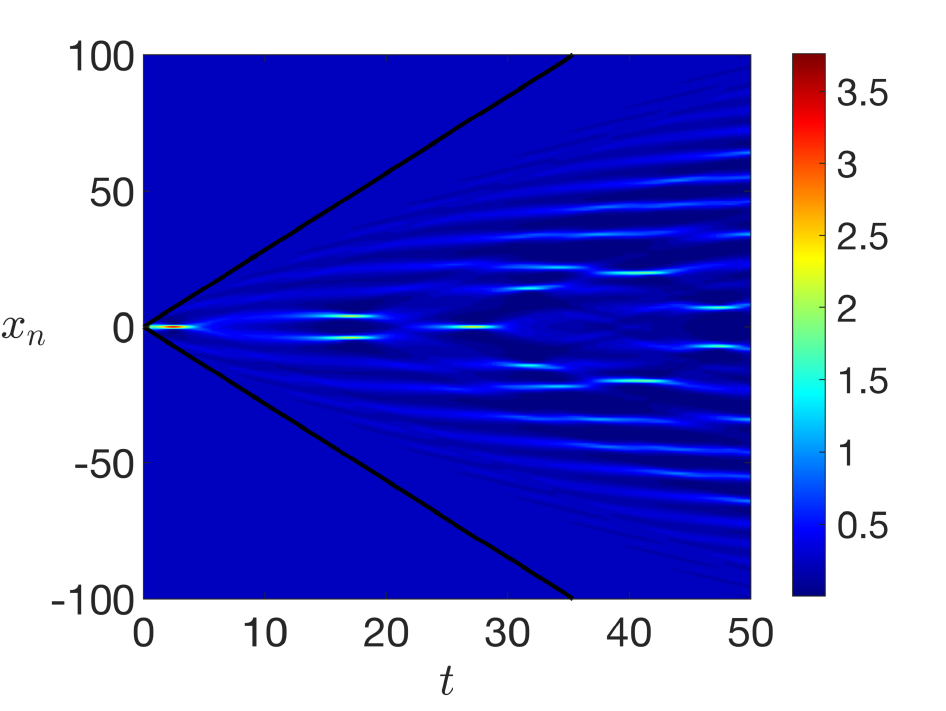}&
		\hspace{0.3cm}\includegraphics[scale=0.4]{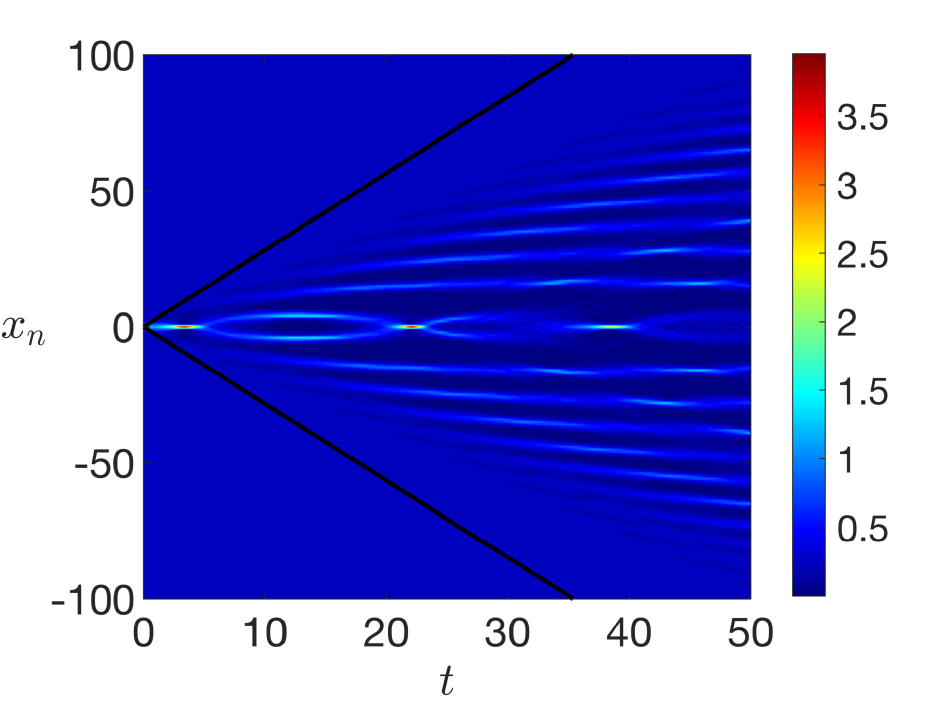}\\[5pt]
		(c)&\hspace{0.7cm}(d)\\
		\includegraphics[scale=0.4]{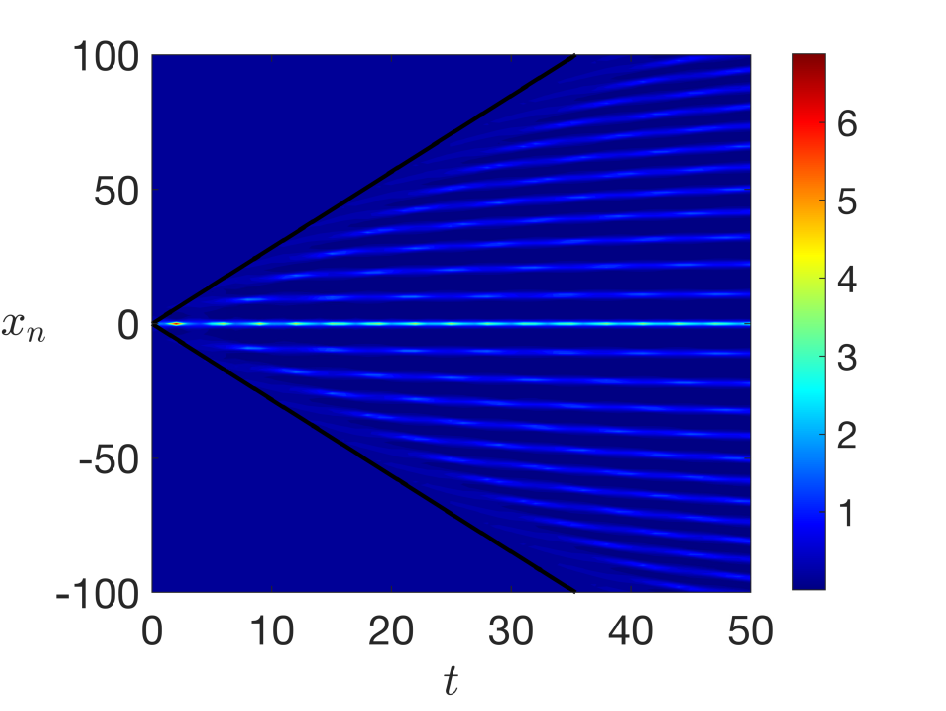}&
		\hspace{0.3cm}\includegraphics[scale=0.4]{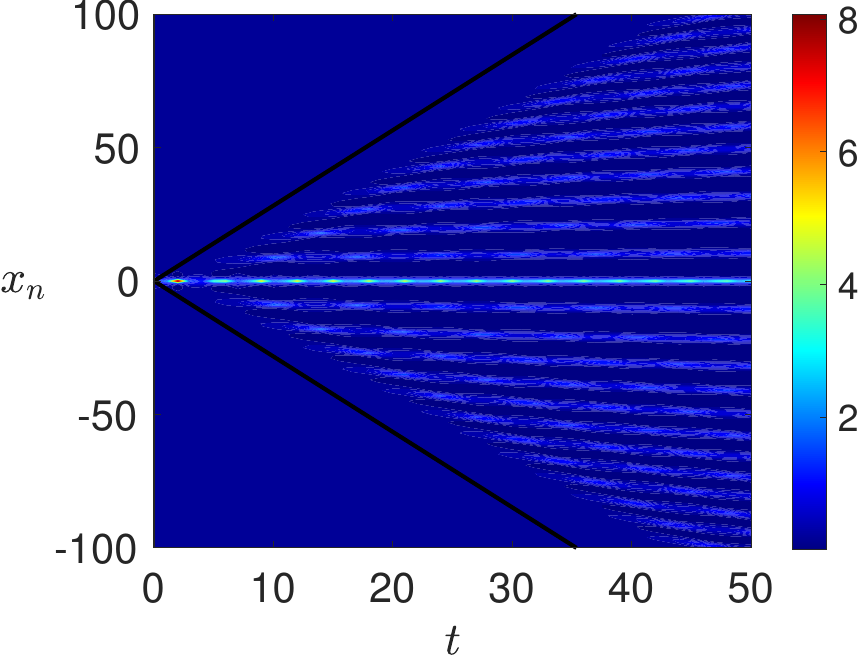}
	\end{tabular}
	\caption{
	\label{Fig9A} Top row: Spatiotemporal evolution of the density $|u_n(t)|^2$ for the DNLS equation \eqref{dnls_gl}, with initial conditions \eqref{loc_ic} for $A=A_*=0.5$  when $\gamma=0.0025$ and $\delta=-0.01$. Panel (a) for $f(x)$ \eqref{alg1} with $\lambda_1=\lambda_2=1$ and $\lambda_3=4$. Panel (b) for  $f(x)$ \eqref{expini} with $\sigma=0.6$, $\rho=1$. Bottom row: The spatiotemporal evolution of the density $|\phi_n(t)|^2$ of the solutions of the AL-lattice \eqref{AL_eq}. Panel (c) for  $f(x)$ \eqref{alg1} and panel (d) for $f(x)$ \eqref{expini} with parameters as above. Rest of the lattice parameters for both systems: $L=200$, $h=1$.   The straight (black) lines depict the graphs of the lines $x=\pm 4\sqrt{2} A t$. More details in the text. }
\end{figure}
\begin{figure}[tbp!]
	\centering 
	\begin{tabular}{cc}
		(a)&\hspace{0.7cm}(b)\\	
		\includegraphics[scale=0.4]{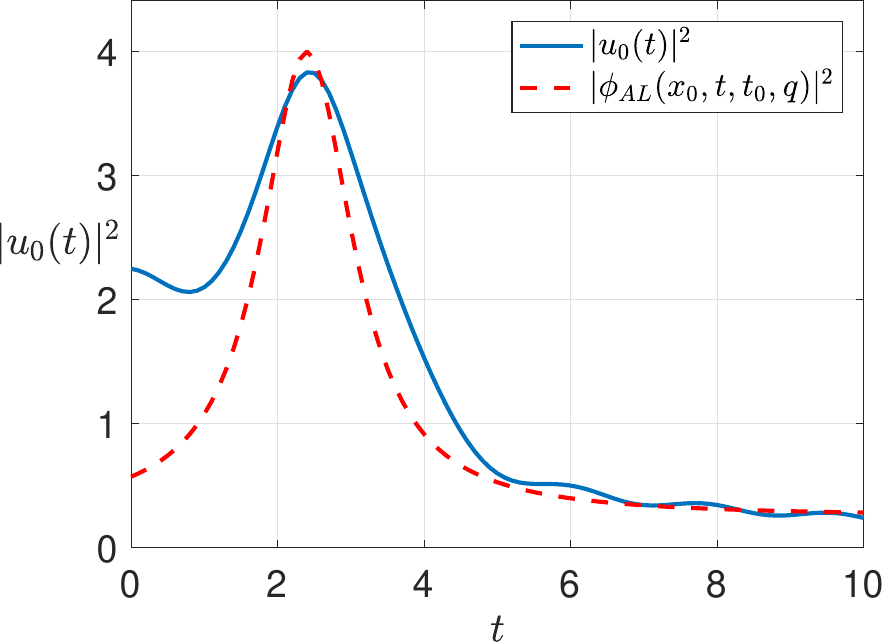}&
		\hspace{0.3cm}\includegraphics[scale=0.4]{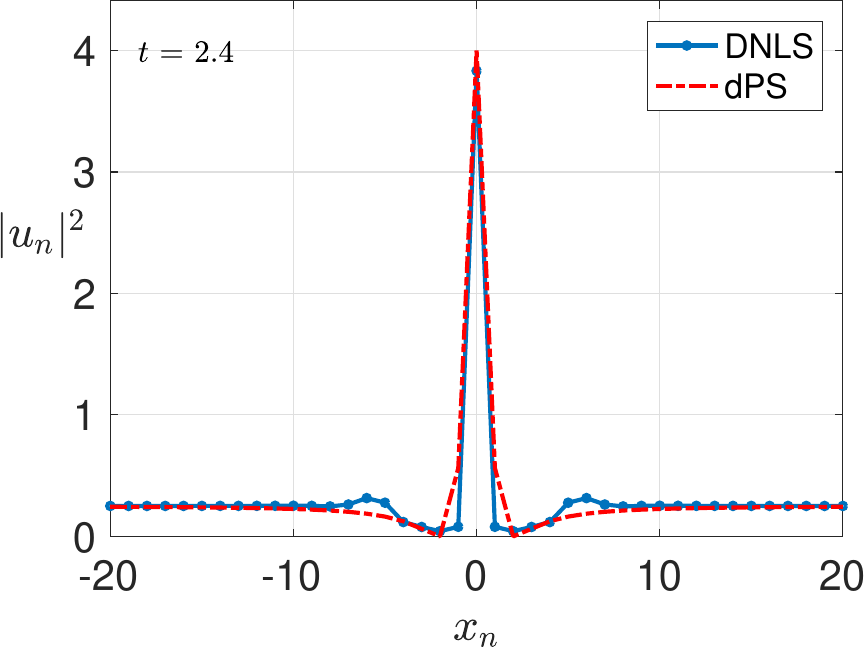}\\[5pt]
		(c)&\hspace{0.7cm}(d)\\
		\includegraphics[scale=0.4]{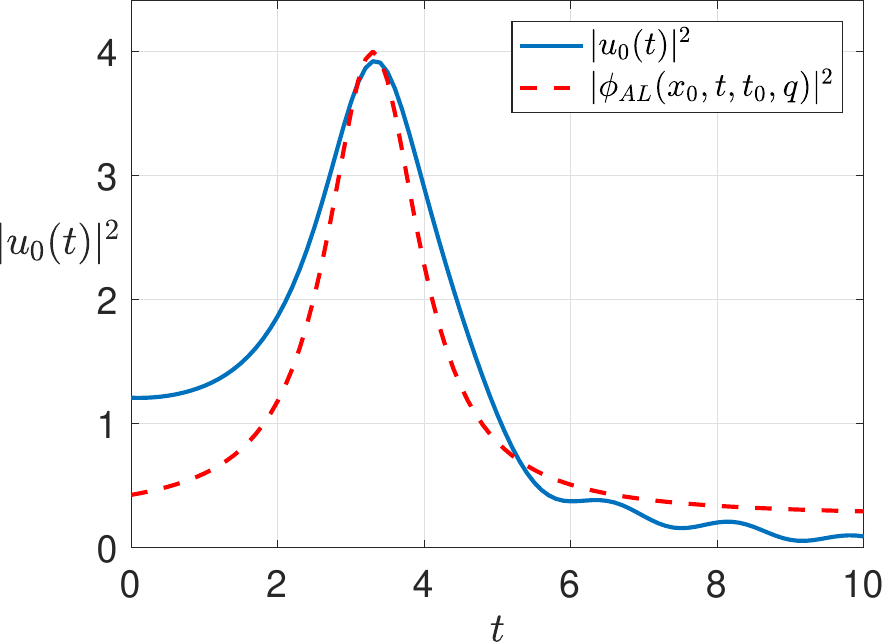}&
		\hspace{0.3cm}\includegraphics[scale=0.4]{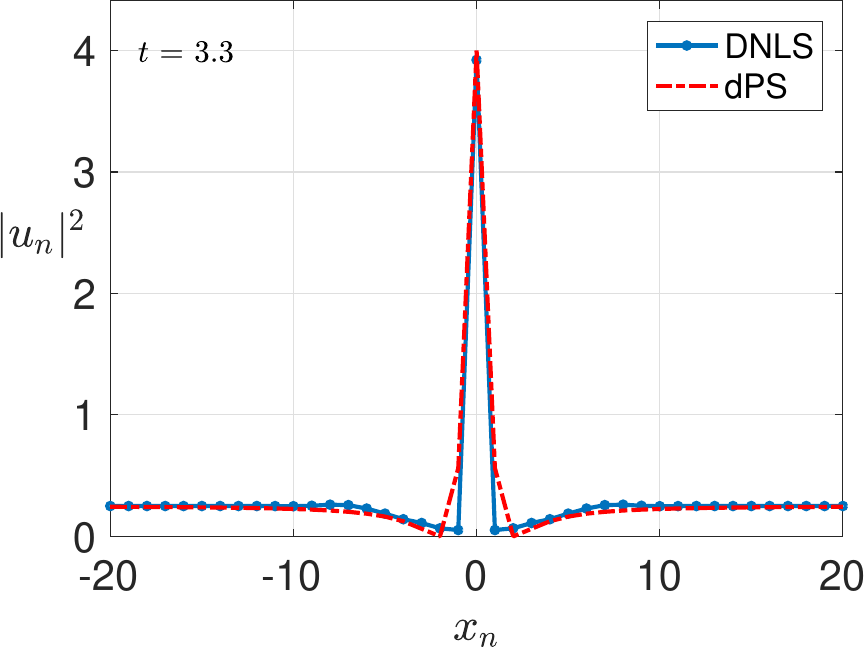}
	\end{tabular}
	\caption{
		\label{Fig9B} Parameters of the lattice  and initial conditions  as in Figure \ref{Fig9A}. Top row: initial condition \eqref{loc_ic} with $f(x)$  \eqref{alg1}. The continuous (blue) curve in panel (a) traces the temporal evolution of the density  $|u_0(t)|^2$  of the central node of the DNLS \eqref{dnls_gl}.  The dashed (red) curve shows the same evolution for the density of the central node  $|\phi_{AL}(x_0,t,2.40,A_*)|^2$ of the analytical dPS solution \eqref{AL}.  Panel (b) depicts a comparison of the profile of the first extreme wave event for the DNLS \eqref{dnls_gl} (see panel (a) of Figure \ref{Fig9A}), against the profile of the analytical dPS solution \eqref{AL} with $q=A_*$, $t_0=2.40$. Bottom row:   Same comparisons as above but for the initial condition \eqref{loc_ic} with $f(x)$  \eqref{expini}.  In this case, the compared dPS solution is  $\phi_{AL}(x_0,t,3.30,A_*)$. More details in the text. }
\end{figure}
\begin{figure}[tbp!]
	\centering 
	\begin{tabular}{cc}
		(a)&\hspace{1cm}(b)\\		\includegraphics[scale=0.4]{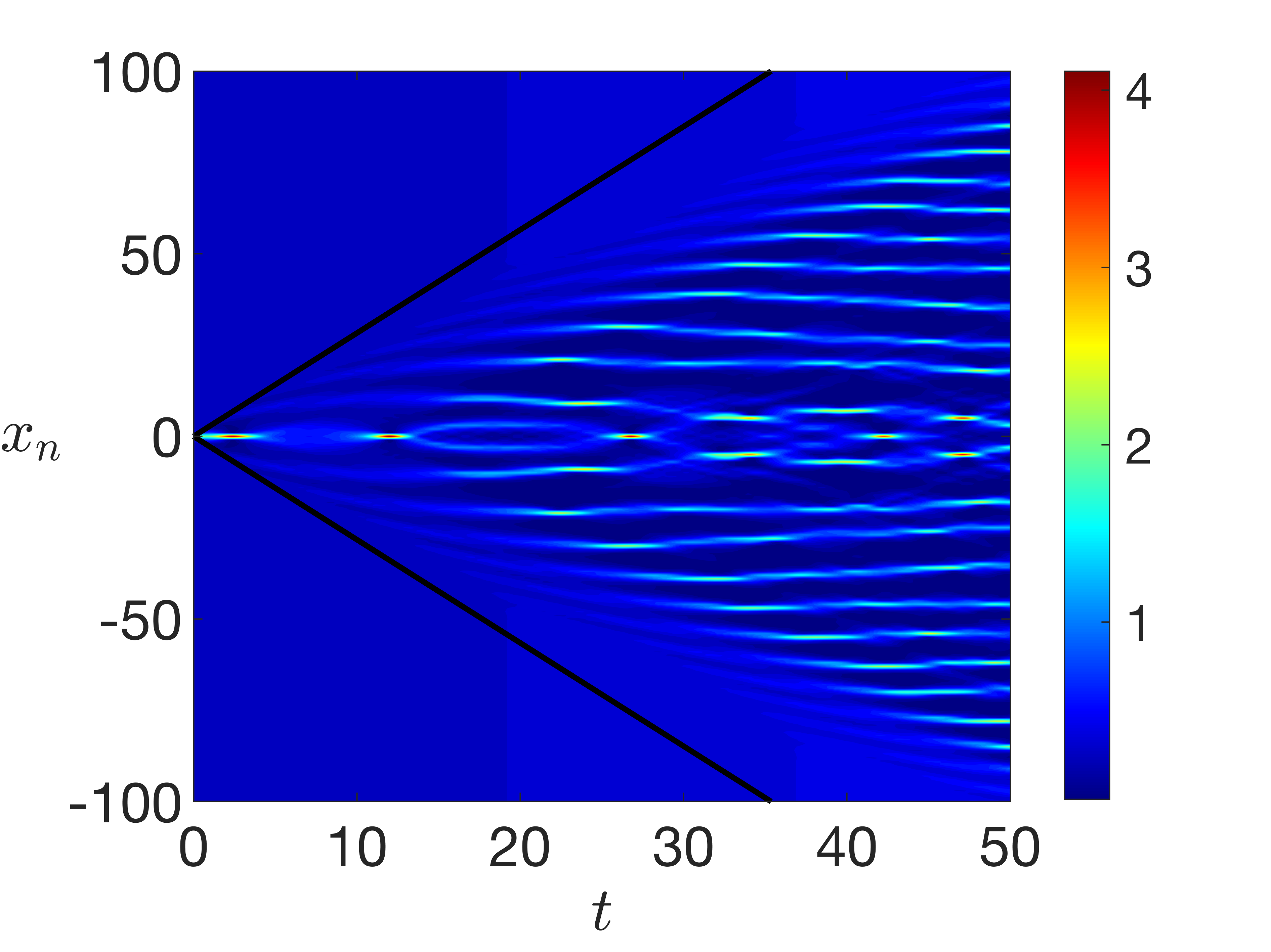}
		&
		\hspace{1cm}\includegraphics[scale=0.4]{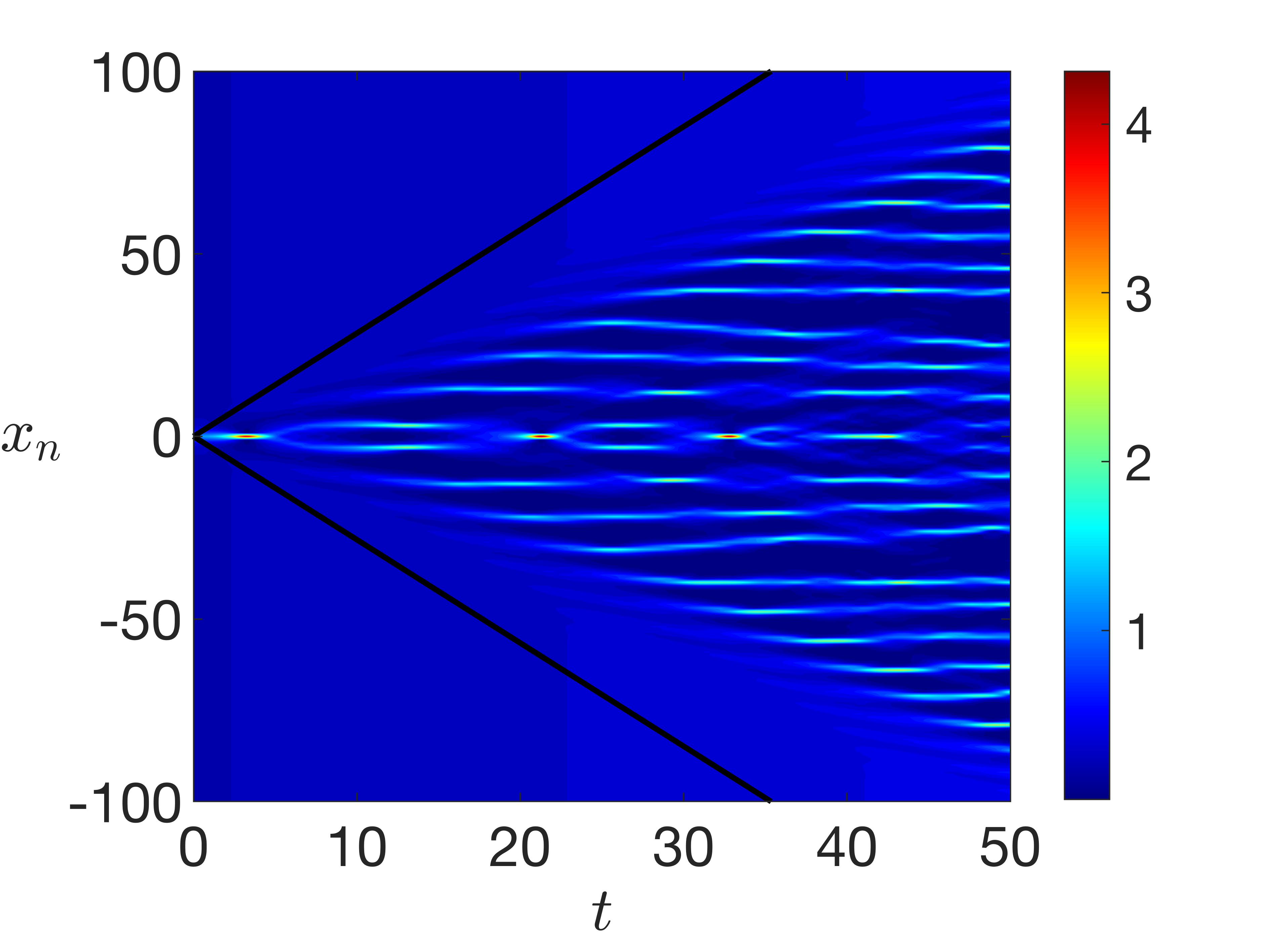}
	\end{tabular}
	\caption{Spatiotemporal evolution of the density $|u_n(t)|^2$ for the DNLS equation \eqref{dnls_gl} with initial conditions \eqref{loc_ic} with $A=0.5$ when $\gamma=0.01$, $\delta=-0.01$, i.e., $A\neq A_*$. Panel (a) for $f(x)$ \eqref{alg1} with $\lambda_1=\lambda_2=1$ and $\lambda_3=4$. Panel (b) for  $f(x)$ \eqref{expini} with $\sigma=0.6$, $\rho=1$.  Rest of the lattice parameters for both systems: $L=200$, $h=1$. The straight (black) lines depict the graphs of the lines $x=\pm 4\sqrt{2} A t$. More details in the text.}
	\label{Fig10A}
\end{figure}

The top row of Figure \ref{Fig9A} depicts the contour plots of the spatiotemporal evolution of the density $|u_n(t)|^2$ of the DNLS \eqref{dnls_gl}, for the initial conditions \eqref{loc_ic}, with $A=A_*$. Panel (a) depicts the dynamics for the quadratically decaying $f(x)$ \eqref{alg1} with $\lambda_1=\lambda_2=1$ and $\lambda_3=4$.  The choice of $\lambda_3=4$ is made so that the initial condition has the same spatial decaying rate as the discrete Peregrine soliton (dPS) \cite{APS}
\begin{equation}\label{AL}
\phi_{AL}(x_n, t, t_{0}, q) = q\Big( 1-\frac{4(1+q^2)(1+4\mathrm{i}q^2(t-t_0 ))}{1+4x_n^2q^2  + 16q^4 (1+q^2)(t-t_0 )^2}\Big) \exp\big(2\mathrm{i}q^2 (t-t_0 )\big),
\end{equation}
the rational analytical solution of the {\em integrable}
AL lattice
\begin{equation}\label{AL_eq}
	\rmi \dot{\phi}_n + k(\phi_{n+1} -2 \phi_n+\phi_{n-1}) + |\phi_n|^2 (\phi_{n-1} + \phi_{n+1}) = 0,
\end{equation}
in the case where $k=1$. We remark that this choice for $\lambda_i$ is not restrictive, since similar dynamics exhibited for various cases of $\lambda_i$. Actually, regarding the general structure of the pattern, the localization rate of the initial condition is proved to be irrelevant as it is shown in panel (b), where apart of the differences of the centered oscillations, almost  the same pattern is structured for the exponentially decaying initial condition $f(x)$ \eqref{expini}, with $\sigma=0.6$ and $\rho=1$. 

The patterns of the DNLS \eqref{dnls_gl} share major characteristics of the ones emerging in the dynamics the AL- lattice when the same initial conditions are used, as it is shown in the bottom row of Figure \ref{Fig9A}, where panel (c) corresponds to the initial condition with $f(x)$ \eqref{alg1} and panel (d) to the initial condition with $f(x)$ \eqref{expini}. In particular, the patterns of the AL system are specific examples, here in the case of the integrable AL-lattice, of the {\em universal behavior  of modulationally unstable media} \cite{blmt2018}. Let us recall that in \cite{blmt2018},  the long-standing open
question about the nonlinear stage of MI on the infinite line for the  integrable NLS equation 
\begin{equation}\label{NLS}
	\rmi u_t + u_{xx}+ |u|^{2}u=0,
\end{equation}
with non-zero boundary conditions, 
\begin{equation}\label{nvNLS}
	\lim_{|x|\rightarrow\infty} u (x,t) =A\exp(\rmi A^2t),\quad A>0
\end{equation}
was fully resolved. It was proved in \cite{bm2017}, that for {\em any given} initial condition
$u(x, 0)$ representing a sufficiently localized initial perturbation of the constant background
$A$,  the solution $u(x, t)$ of \eqref{NLS} is
$u(x, t) = u_{asymp}(x, t) + \mathrm{O}(1/\sqrt{t})$ as $t\rightarrow\infty$. Cruciallly, $u_{asymp}(x, t)$ has different forms
in different sectors of the $xt$-plane, namely: (i) in the two “plane-wave” regions,
$x < -4\sqrt{2}At$ and $x>4\sqrt{2}At$, 
we have that $|u_{asymp}(x, t)| = A$, that is, the solution has the
same amplitude as the {\em undisturbed} background $A$ (ii) in the “modulated elliptic wave” region $ -4\sqrt{2}At<x<4\sqrt{2}At$
the solution is expressed by a slow modulation of the elliptic solutions of \eqref{NLS}. In \cite{blmt2018}, analytical arguments corroborated with numerical simulations provided evidence that this behavior is universal, as it is shared by a wide class of NLS-type equations and systems, including the AL-lattice \eqref{AL_eq} and the Hamiltonian DNLS, which corresponds to the case of $\gamma=\delta=0$ of the dissipative DNLS \eqref{dnls_gl} (see \cite[Appendix C, pg. 903]{blmt2018}). The findings of \cite{blmt2018} regarding the universal behavior of MI for the DNLS systems, combined with the theoretical results proven herein, will serve as a roadmap for explaining our numerical results, as elaborated below. 

In this context, of the universality of the MI behavior, we claim that the findings depicted for the DNLS \eqref{dnls_gl} in Figure \ref{Fig9A} are important. First we should stress that the DNLS \eqref{dnls_gl} is a system at least ``two steps forward" regarding the breaking of the integrability barrier defined by the AL system, since it is a dissipative perturbation of the non-integrable Hamiltonian DNLS. Second, the theoretical results of Theorems \ref{TH1} and \ref{T1} accompanied with the MI analysis of the plane-wave of amplitude $A_*$ are fully explaining the dynamics observed in panels (a) and (b) if combined with the findings of \cite{blmt2018} mentioned above. The illustrated dynamics show that the breaking of integrability induced by the DNLS \eqref{dnls_gl} is not dramatic. It shares the major characteristics of the dynamics of the  AL-lattice and of the Hamiltonian DNLS observed in \cite{blmt2018}, that is, the structure comprised of  the two outer, quiescent sectors separated by the wedge-shaped central region characterized by the oscillatory behavior. The straight (black) lines depict the graphs of the lines $x=\pm 4\sqrt{2} A t$, which are the exact boundaries of the wedge in the case of the integrable NLS partial differential equation, calculated in \cite{bm2017}. It is crucial to remark that, according to the analysis of \cite{blmt2018}, the slopes of the boundaries are modified depending on the particular system considered; this modification also occurs in the case of the dissipative DNLS lattice \eqref{dnls_gl}. It is interesting to observe, however, that the slopes of the lines $x=\pm 4\sqrt{2} A t$ approximate quite well the slopes of the boundaries of the wedge in the case of the AL lattice (panels (c) and (d)). It could be of interest to investigate if this fact could be another manifestation of the integrability of the AL lattice despite its discrete nature. In the outer sectors, the amplitude of $u_n(t)$ is $A_*$, of the undisturbed background, the only one which may sustain localized solutions in the case of the infinite lattice. On the other hand, since the numerical solutions are generated by solving numerically the periodic lattice, this background will never lose stability in terms of the amplitude in the outer regions since the solutions are initiated on the stable attractor of the finite system regarding its amplitude. Yet for the finite dimensional system, the MI dynamics in the wedge-shaped central region can be only transient as it is proved in Theorem \ref{T1} and analyzed in the discussion of section  \ref{subsec1}. Importantly, in the case of the finite lattice the wedge will eventually disappear as the system will select a stable wave number defining the global attractor, and the asymptotic state of the solution will be the plane wave of amplitude $A_*$ with the stable wavenumber. However, in subsections C and D that follow, we provide estimates for the $\ell^2_{per}$- distance between the solutions of the AL lattice and the DNLS \eqref{dnls_gl}. As we will discuss therein, these estimates provide additional theoretical justification for the proximity of the dynamics between the AL lattice and the DNLS \eqref{dnls_gl}, {\em at least in the early stages of the MI behavior.}

Focusing in shorter time intervals, where the formation of the first events occur, it is interesting to observe that these are reminiscent of the dPS. This similarity is illustrated in Figure \ref{Fig9B}. In this figure, a comparison of the dynamics of the DNLS \eqref{dnls_gl}  against the analytical dPS solution \eqref{AL} of the AL-lattice \eqref{AL_eq} is made, for short times around the emergence of the first extreme events spotted in panels (a) and (b) of Figure \ref{Fig9A}.  Panels (a) and (b) correspond to the case of the initial condition \eqref{loc_ic} with the quadratically decaying $f(x)$ \eqref{alg1}, and panels (c) and (d) to the case of the sech-type $f(x)$ \eqref{expini}. Panels (a) and (c) depict the temporal evolution of the density of the central node $|u_0(t)|^2$  of DNLS \eqref{dnls_gl}, traced by the continuous (blue) curve.   The dashed (red) curves depict the same evolution for the density of the central node  $|\phi_{AL}(x_0,t, t_0,A_*)|^2$ of the analytical dPS solution \eqref{AL}.  For the algebraic $f(x)$ \eqref{alg1}, the time of occurrence of the rogue wave is $t_0=2.40$ and for the sech-type $f(x)$ \eqref{expini}, the time is $t_0=3.30$. These instants $t_0$ are used in the formula \eqref{AL} in order to define the dPS solutions which are used for the comparison.  We observe the remarkable similarity of the time growth and decay rates of the events, close to that of the analytical dPS solutions.   Panels (b) and (d) depict a comparison of the profiles of the first extreme wave events for the DNLS \eqref{dnls_gl}, against the profiles of the analytical dPS solution \eqref{AL} with $q=A_*$, and the corresponding times $t_0$ mentioned above. Again, the similarity of the profiles is remarkable, and importantly, the coincidence of the supporting backgrounds of the solutions $A=A_*$. It is interesting to observe that the profile emerging from the $\mathrm{sech}$-profiled initial condition is more proximal to the dPS analytical profile. 
\begin{figure}[tbp!]
	\centering 
	\begin{tabular}{cc}
		(a)&\hspace{0.7cm}(b)\\	
		\includegraphics[scale=0.4]{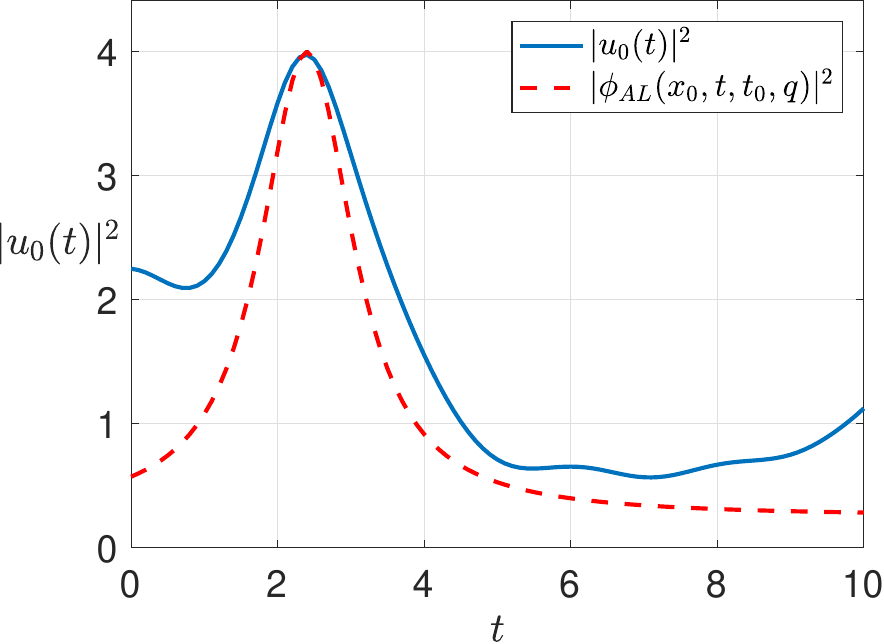}&
		\hspace{0.3cm}\includegraphics[scale=0.4]{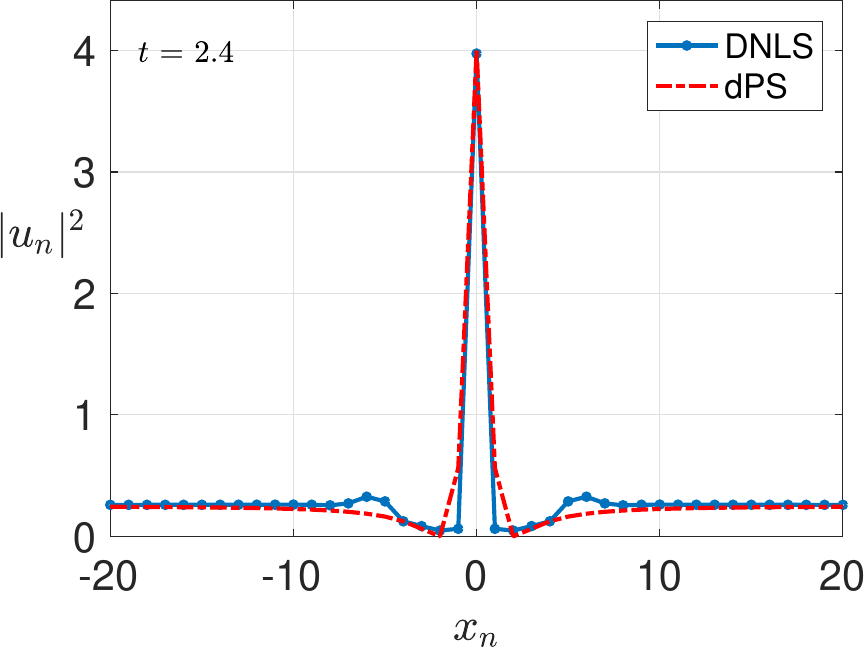}\\[5pt]
		(c)&\hspace{0.7cm}(d)\\
		\includegraphics[scale=0.4]{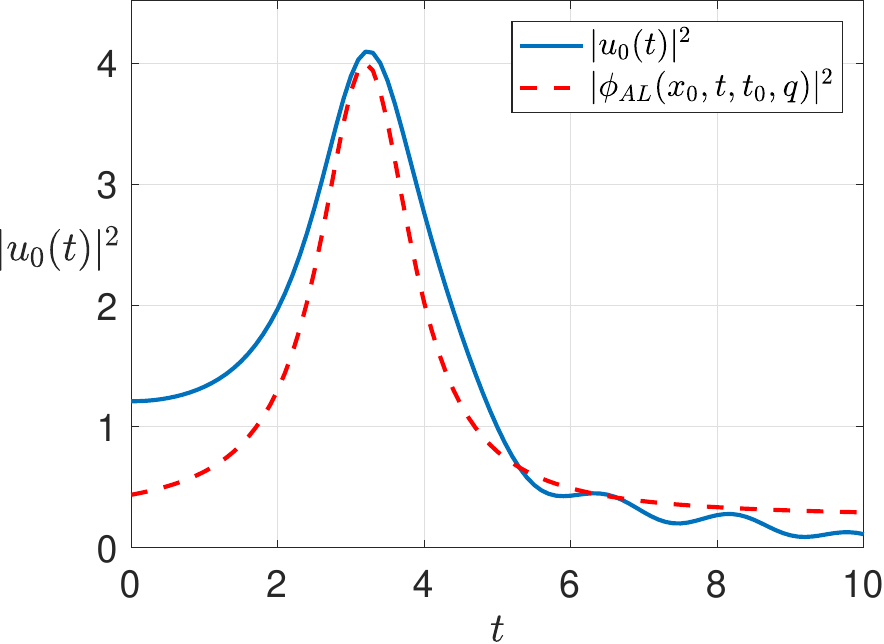}&
		\hspace{0.3cm}\includegraphics[scale=0.4]{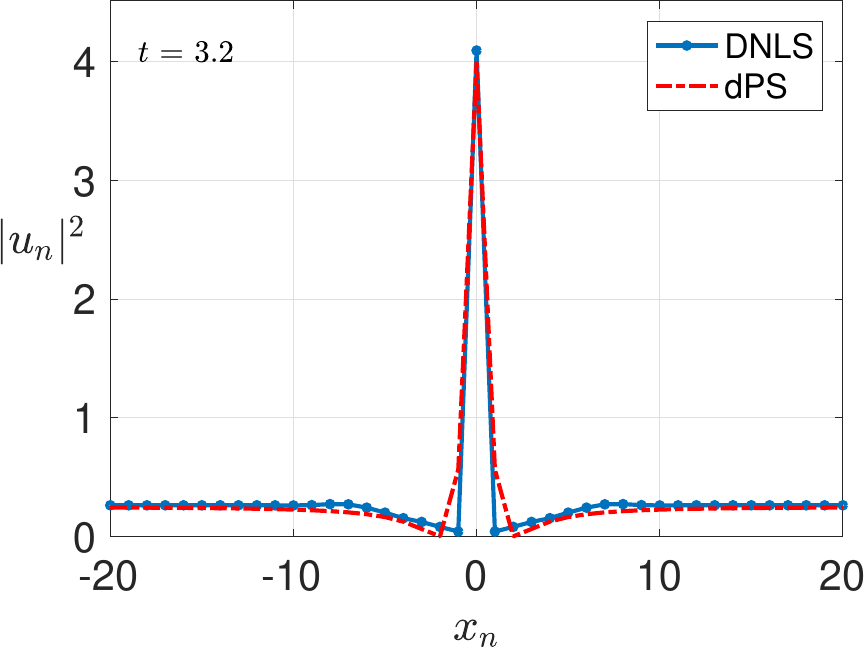}
	\end{tabular}
	\caption{
		\label{Fig11A} The counterpart of the study of Figure \ref{Fig9B} in the case $A\neq A_*$. Parameters of the DNLS lattice \eqref{dnls_gl} and initial conditions  as in Figure \ref{Fig10A}. Details in the text. }
\end{figure}
\paragraph{The case $A\neq A_*$.}  The parameters and initial conditions are as in the previous case, with the major difference that $\gamma=-\delta=0.01$ which define the value for $A_*=1$. Therefore, the initial conditions are on top of the background with amplitude  $A=0.5\neq A_*$. {\em It is crucial to remark that the results are not relevant for the infinite lattice with nonzero boundary conditions since according to Theorem \ref{T1}, solutions do not exist when $A\neq A_*$.}   However, we see in Figure \ref{Fig10A}, that for finite times, the characteristic MI pattern is present. The important feature is that since the initial condition is on the top of $A\neq A_*$,  the wedge-shaped region is not formed on an undisturbed plane wave region. In full agreement with the theoretical analysis, since $A<A^*$, the wedge is developed on a background of increasing amplitude according to the ode solution  \eqref{eqf}; this is evident by the change of ''shading'' in the outer regions of the wedge.  The background eventually converges to that of the global attractor $A^*$, accompanied by the disappearance of the wedge, as predicted theoretically by Theorem \ref{TH1}. 

However, since we have chosen $A=0.5<A_*=1$, for short times prior those for which the effect of the increasing amplitude will become significant, we expect a similarity with the dynamics illustrated in Figure \ref{Fig9B} for the case $A=0.5=A_*$. This fact is shown in Figure \ref{Fig11A} which illustrates the same comparison of the first emerging extreme wave events spotted in Figure \ref{Fig10A} against the analytical dPS solutions. The dynamics are found to be almost indistinguishable to the one presented in Figure \ref{Fig9B}. Nevertheless, when examining the dynamics for larger times, we can clearly see a distinction between the scenario where  $A=0.5=A_*$ and the scenario where $A=0.5\neq A_*=1$. To illustrate this, we focus on the sech-type initial conditions \eqref{expini} and refer to Figures \ref{Fig9A}b and \ref{Fig10A}b. Panel (a) of Figure \ref{Fig_new} shows a snapshot of the solution at $t=38.5$ when $A=0.5=A_*$. It is evident that the amplitude of the background $A_*$ is preserved in full accordance with Theorem \ref{TH1}.  In contrast, in panel (b) of Figure \ref{Fig_new} which corresponds to the case $A=0.5\neq A_*=1$ and shows a snapshot of the solution at $t=32.8$, the background of the solution is elevated, in accordance with Theorem \ref{T1}.

Thus, we may expect that two systems with initial conditions \eqref{loc_ic} with  $A_1=A_2$  may exhibit similar dynamics for short time intervals when $A_*=A_1$ and $A_*\neq A_2$ which defer for large times, particularly regarding their long-time asymptotics. Accordingly, we may expect a similarity of the dynamics of both of the above systems to the dynamics of the AL-lattice for short-time intervals. We will discuss this issue theoretically, in the next section.

\begin{figure}[tbp!]
	\centering 
	\begin{tabular}{cc}
		(a)&\hspace{0.7cm} (b)\\
		\includegraphics[scale=0.4]{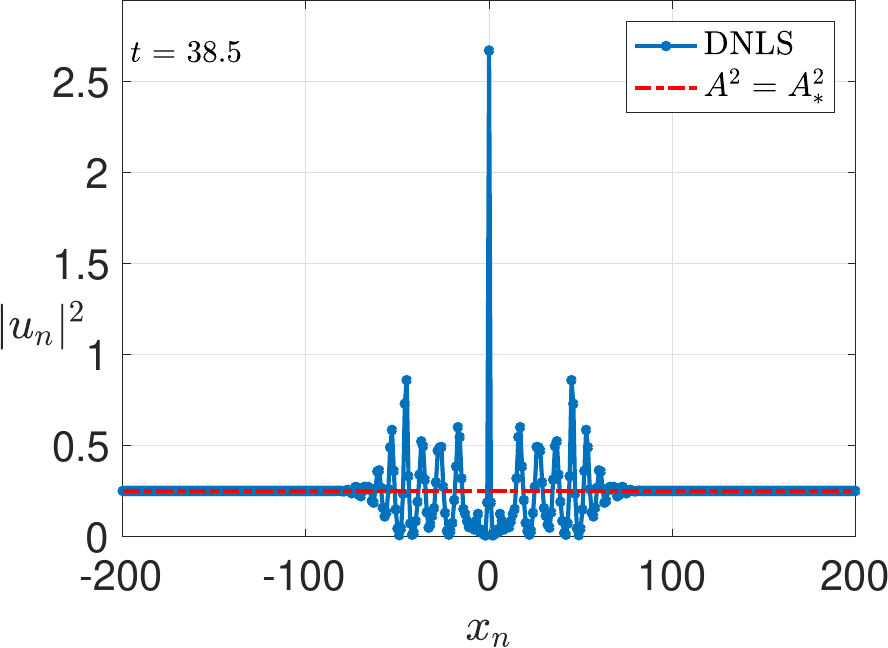}&
		\hspace{0.3cm}
		\includegraphics[scale=0.4]{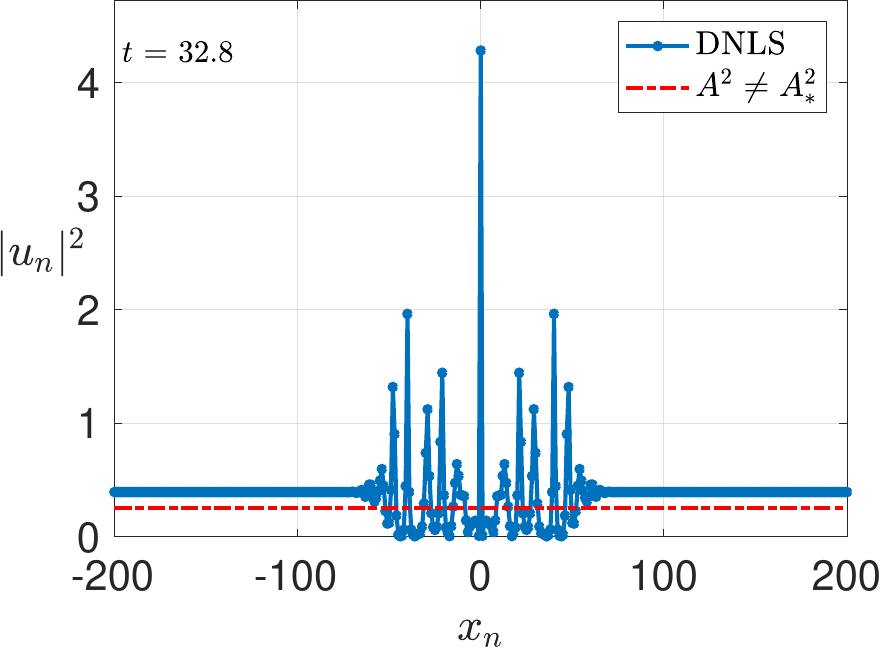}
	\end{tabular}
	\caption{
		\label{Fig_new} The profile of the density $|u_n(t)|^2$ (blue curve) of the third extreme event of the DNLS \eqref{dnls_gl} with  initial conditions \eqref{loc_ic} and $f(x)$ \eqref{expini} with $\sigma=0.6$, $\rho=1$.  The dashed (red) curve shows the background of the initial condition for  $A=0.5$.  Panel (a): $A_*= A$ and $t=38.5$.  Panel (b): $A_*= 1$ and $t=32.8$. }
\end{figure}
\subsection{Return to the analytical considerations: measuring the distance between the solution the DNLS \eqref{dnls_gl} and of the Ablowitz-Ladik lattice}
The following results provide a theoretical consideration of the numerical observations on the distance between the analytical rogue wave solution of the AL-lattice and the rogue wave alike structure emerged as a first event in the dynamics of the DNLS \eqref{dnls_gl}.  First, we prove the following property of the solutions of the AL-lattice \eqref{AL_eq} when supplemented with the periodic boundary conditions \eqref{eq02}. 
\begin{lemma}
	\label{ALper}
	For the initial condition $\phi_n(0)\in\ell^2_{per}$ of the AL-lattice \eqref{AL_eq}, we consider the quantity
	\begin{equation}
		\label{aux0}
		\mathcal{N}(0)=h\sum_{n=0}^{N-1}\ln(1+|\phi_n(0)|^2)< \infty.
	\end{equation}
	Then, the corresponding solution of the AL-lattice satisfies the estimate

	\begin{equation}
		\| \phi(t) \|_{\ell^2_{per}}^2 \le h\exp(h^{-1}\mathcal{N}(0))-h,\qquad \forall t\ge 0.\label{nAL}
	\end{equation}
\end{lemma}
\begin{proof} The fact that $\mathcal{N}(0)$ is finite follows by the elementary inequality $\ln(1+x)\leq x$ for all $x>0$ which implies that
	\begin{equation}
		\label{AL01}
		\mathcal{N}(0)\leq h\sum_{n=0}^{N-1}|\phi_n(0)|^2=\|\phi(0)\|^2_{\ell^2_{per}}
	\end{equation}
Setting $\lambda_n^2(t)=\ln(1+|\phi_n(t)|^2)$, we write $\mathcal{N}(t)$ in the form
	\begin{equation}
		\label{auxeq1}
		\mathcal{N}(t)= h\sum_{n=0}^{N-1}\ln(1+|\phi_n(t)|^2)=h\sum_{n=0}^{N-1}\lambda^2_n(t).
	\end{equation}
	We recall that $\mathcal{N}(t)$ is one of the conserved quantities of the AL-lattice and has this property still in the case of the periodic boundary conditions \eqref{eq02}, \cite{PM1}.  From \eqref{auxeq1}, and by using the equivalence of norms \eqref{equi} for $q=2j$ and $p=2$, which implies that  
	\begin{equation*}
		h\sum_{n=0}^{N-1}|\lambda_n|^{2j}\leq h^{1-j} \bigg(h\sum_{n=0}^{N-1}|\lambda_n|^{2}\bigg)^j,
	\end{equation*}
	we get the estimate:
	\begin{align}
		h\sum_{n=0}^{N-1}|\phi_n|^2&=h\sum_{n=0}^{N-1}\big(\exp(|\lambda_n|^2)-1\big)=\sum_{n=0}^{N-1}\,\bigg(\sum_{j=0}^{\infty}\frac{|\lambda_n|^{2j}}{j!}-1\bigg)\nonumber\\
		&=h\sum_{n=0}^{N-1}\,\sum_{j=1}^{\infty}\frac{|\lambda_n|^{2j}}{j!}=\sum_{j=1}^{\infty}\frac{h}{j!}\,\sum_{n=0}^{N-1}|\lambda_n|^{2j}\nonumber\\
		&\le\sum_{j=1}^{\infty}\frac{h^{1-j}}{j!}\,\bigg(\sum_{n=0}^{N-1}h|\lambda_n|^{2}\bigg)^j=h\sum_{j=1}^{\infty}\frac{\left(h^{-1}\mathcal{N}\right)^j}{j!}\nonumber\\
		&=h\bigg(\sum_{j=0}^{\infty}\frac{\left(h^{-1}\mathcal{N}\right)^j}{j!}-1\bigg)=h\exp(h^{-1}\mathcal{N})-h.
	\end{align}
	From the conservation of $\mathcal{N}$, that is $\mathcal{N}(t)=\mathcal{N}(0)$ for all $t\geq 0$, it follows that
	\begin{equation*}
		h\sum_{n=0}^{N-1}|\phi_n(t)|^2\le h\exp(h^{-1}\mathcal{N}(0))-h,\qquad \forall t\ge 0,
	\end{equation*}
	which is the claimed estimate \eqref{nAL}.
\end{proof}

\begin{theorem}
	\label{proximityT}
	Let $u_n(0),\;\phi_{n}(0)\in\ell^2_{per}$ be the initial conditions of the DNLS \eqref{dnls_gl} and of the  AL-lattice \eqref{AL_eq}, respectively.  For the initial condition of the  DNLS \eqref{dnls_gl} we assume additionally that 
	\begin{equation}
		\label{estini}
		P_a[u(0)]:=\frac{1}{N}\sum_{n=0}^{N-1}|u_n(0)|^2< A_*^2.
	\end{equation}
	Then, the distance $\|Y(t)\|_{\ell^2_{per}}=\|u(t)-\phi(t)\|_{\ell^2_{per}}$  satisfies for all $t\geq 0$, the following estimates:
	\begin{itemize}
		\item {\em estimate I:} 
		\begin{equation}
			\label{proxest1}
			\|u(t)-\phi(t)\|_{\ell^2_{per}}\leq  \| u(0)-\phi(0)\|_{\ell^2_{per}} + \gamma F_1(t)+\big(\sqrt{\delta^2 + 1}\big)F_2(t)+2\left[\exp(\mathcal{N}(0))-1\right]^{3/2}t:=\mathcal{F}(t),\quad
		\end{equation}
		where the functions $F_1(t),F_2(t)$ are defined explicitly as
		\begin{align*}
			F_1(t)&=\int_{0}^{t}\sqrt{B(s)}ds=\frac{1}{\sqrt{\beta\gamma}}\ln\Big(\frac{\sqrt{\gamma\nu}-\sqrt{\beta}}{\sqrt{\beta}e^{-\gamma t}+\sqrt{(e^{2\gamma t}-1)\beta+\gamma\nu}}\Big),\\
			F_2(t)&=\int_{0}^{t}B(s)^{3/2}ds\nonumber\\
			&=\frac{1}{\beta^{3/2}}\bigg[\sqrt{\gamma}\mathrm{ArcSinh}\Big(\frac{\sqrt{\beta}e^{\gamma t}}{\sqrt{\nu\gamma-\beta}}\Big)-\frac{\sqrt{\beta\gamma}e^{\gamma t}}{\sqrt{\beta(e^{2\gamma t}-1)+\gamma\nu}}-\sqrt{\gamma}\mathrm{ArcSinh}\Big(\frac{\sqrt{\beta}}{\sqrt{\nu\gamma-\beta}}\Big)+\frac{\sqrt{\beta}}{\sqrt{\nu}}\bigg],\\
			\nu&=\|u(0)\|_{\ell^2_{per}}^{-2},\quad \beta=\frac{\tilde{\delta}}{Nh}=-\frac{\delta}{Nh}>0.
		\end{align*}
		\item {\em estimate II:} 
		\begin{align}
			\label{proxest3}
			\|u(t)-\phi(t)\|_{\ell^2_{per}}&\leq \| u(0)-\phi(0)\|_{\ell^2_{per}}+ \alpha t:=\mathcal{F}_b(t),\\
			\label{proxest4}
			\alpha &=	\gamma A_*\sqrt{Nh} + \big(\sqrt{\delta^2 + 1}\big)\sqrt{h}A_*^3N^{3/2} + 	2\sqrt{h}\left[\exp(h^{-1}\mathcal{N}(0))-1\right]^{3/2}	.
		\end{align}
	\end{itemize}
\end{theorem}	
\begin{proof}
	We derive the equation for $Y_n(t)=u_n(t)-\phi_n(t)$, by subtracting the equations \eqref{dnls_gl}-\eqref{AL_eq}. Then, $Y_n(t)$ satisfies 

	\begin{equation}\label{main}
		\rmi\dot{Y}_n +  \Delta_d Y_n =  \rmi \gamma u_n + (\rmi\delta - 1) |u_n|^2 u_n  +  |\phi_n|^2 (\phi_{n-1} + \phi_{n+1}).
	\end{equation}
	We multiply the equation \eqref{main} by $\overline{Y}_n$ in the $\ell^2_{per}$-inner product.  By keeping the imaginary  parts of the resulting equation, we get the balance law for $Y$,
	\begin{equation}\label{im_eq}
		\frac{1}{2} \frac{d}{dt} \|Y\|_{\ell^2_{per}}^2  =
		\gamma h\sum_{n=0}^{N-1} u_n \overline{Y}_n + h\mathrm{Im}\Big\{ (\rmi\delta - 1)\sum_{n=0}^{N-1} |u_n|^2 u_n \overline{Y}_n  \Big\} + h\,\mathrm{Im} \Big\{\sum_{n=0}^{N-1} |\phi_n|^2 (\phi_{n-1} + \phi_{n+1})\overline{Y}_n  \Big\}.  
	\end{equation}
	Using the  Cauchy--Schwarz inequality, the first term  at the right-hand side of \eqref{im_eq} can be estimated as
	\begin{equation}
		\label{dist1}
		h\sum_{n=0}^{N-1} u_n \overline{Y}_n=\sum_{n=0}^{N-1} \sqrt{h}u_n\sqrt{h} \overline{Y}_n  \le  \|u\|_{\ell^2_{per}}\|Y\|_{\ell^2_{per}}.
	\end{equation}
	For the second  term of the right-hand side of \eqref{im_eq} we have the estimate:
	\begin{align}
		\label{dist2}
		h\mathrm{Im}\Big\{ (\rmi\delta-1)\sum_{n=0}^{N-1} |u_n|^2|u_n| |\overline{Y}_n | \Big\}  & \le  \big(\sqrt{\delta^2 + 1}\big)\big(h^{-1}h\sup_{0\leq n\leq N-1}|u_n|^2\big) h\sum_{n=0}^{N-1}|u_n| |Y_n| \nonumber\\
		&\le  \big(\sqrt{\delta^2 + 1}\big)h^{-1}\|u\|^2_{\ell^2_{per}} h\sum_{n=0}^{N-1}|u_n||Y_n| \\
		&\le \big(\sqrt{\delta^2 + 1}\big)h^{-1}\|u\|_{\ell^2_{per}}^3\|Y\|_{\ell^2_{per}}.\nonumber
	\end{align}
	For the third term, let us note that due to the periodic boundary conditions \eqref{eq02}, we have that
	\begin{equation}
		\label{dist3}
		\sum_{n=0}^{N-1}|\phi_{n-1}|^2=\sum_{n=0}^{N-1}|\phi_{n+1}|^2=\sum_{n=0}^{N-1}|\phi_{n}|^2.
	\end{equation}
	Then,  by using \eqref{dist3}, the third term of the right-hand side of \eqref{im_eq}  can be  estimated as follows:
	\begin{align}
		\label{dist4}
		&h\mathrm{Im}\Big\{ \sum_{n=0}^{N-1}|\phi_n|^2 |\phi_{n-1} + \phi_{n+1}| |\overline{Y}_n| \Big\}  \le  
		\Big(h^{-1}h\sup_{0\leq n\leq N-1}|\phi_n|^2\Big)\Big\{h\sum_{n=0}^{N-1}|\phi_{n-1}|\,|Y_n|+h\sum_{n=0}^{N-1}|\phi_{n+1}|\,|Y_n|\Big\}
		\nonumber\\ 
		& \le h^{-1}\|\phi\|^2_{\ell^2_{per}}\bigg\{\Big(h\sum_{n=0}^{N-1}|\phi_{n-1}|^2\Big)^{1/2}\Big(h\sum_{n=0}^{N-1}|Y_n|^2\Big)^{1/2} + \Big(h\sum_{n=0}^{N-1}|\phi_{n+1}|^2\Big)^{1/2}\Big(h\sum_{n=0}^{N-1}|Y_n|^2\Big)^{1/2}\bigg\}
		\\
		& \le 
		2h^{-1}\|\phi\|_{\ell^2_{per}}^3\|Y\|_{\ell^2_{per}}.\nonumber
	\end{align}
	Now,  using the estimates \eqref{dist1}, \eqref{dist2} and \eqref{dist4} together with the equation
	$$
	\frac{1}{2}\frac{d}{dt}  \|Y\| _{\ell^2_{per}}^2 =  \| Y \| _ {\ell^2_{per}} \frac{d}{dt} \| Y \|_ {\ell^2_{per}} ,
	$$
	we derive from \eqref{im_eq}, the differential inequality
	\begin{align}\label{dist5}
		\frac{d}{dt}\| Y\|_ {\ell^2_{per}}  \le 
		\gamma \|u\|_{\ell^2_{per}} + \big(\sqrt{\delta^2 + 1}\big)h^{-1}\|u\|_{\ell^2_{per}}^3 + 	2h^{-1}\|\phi\|_{\ell^2_{per}}^3.
	\end{align}
	{\em Proof of estimate I:} We bound from above the right hand side of \eqref{dist5} by  using the estimate 
	\begin{equation}\label{eq27l}
		\|u(t)\|^2_{\ell^2_{per}} \leq \frac{\gamma}{\gamma\exp(-2\gamma t)\|u(0)\|^{-2}_{\ell^2_{per}} + \frac{\tilde{\delta}}{Nh}[1-\exp(-2\gamma t)]}:=B(t), \quad  \forall t\geq 0,
	\end{equation}
	(which can be easily deduced for $\|u(t)\|_{\ell^2_{per}}$ from the bound \eqref{eq27}) and the estimate \eqref{nAL} for $\|\phi(t)\|_{\ell^2_{per}}$. We remark that under the assumption \eqref{estini}, the function $B(t)$ satisfies
	
	\begin{equation}
		\label{B1}
		0<B(t)< A_*^2Nh,\quad \mbox{for all $t>0$},
	\end{equation}
	and recall from Theorem \ref{T1}, that 
	\begin{equation}
		\label{B2}
		\lim_{t\rightarrow\infty}B(t)=A^2_*Nh.
	\end{equation}
	Then we integrate in the arbitrary interval $[0,t]$ to get the claimed estimate \eqref{proxest1}. We also note that the assumption \eqref{estini} can be written alternatively as 
	\begin{equation}
		\label{B3}
		\nu\gamma>\beta,
	\end{equation}
	and due to \eqref{B3} the functions $F_1(t)$ and $F_2(t)$ are well defined for all $t>0$.\\
	{\em Proof of the estimate II:} Since the bound \eqref{B1} holds for all $t>0$, the right-hand side of \eqref{dist5} can be estimated alternatively, giving the differential inequality
	\begin{align}\label{dist6}
		\frac{d}{dt}\| Y\|_ {\ell^2_{per}}  \le 
		\gamma A_*\sqrt{Nh} + \big(\sqrt{\delta^2 + 1}\big)\sqrt{h}A_*^3(N)^{3/2} + 	2\sqrt{h}\big[\exp(h^{-1}\mathcal{N}(0))-1\big]^{3/2}.
	\end{align}
	Integration of \eqref{dist6} in the arbitrary interval $[0,t]$ proves the alternative estimate \eqref{proxest3} with the at most linear growth rate \eqref{proxest4}.
\end{proof}

\subsection{Remarks on the numerical observations of section \texorpdfstring{\ref{RWE}}{} under Theorem \texorpdfstring{\ref{proximityT}}.}
The numerical observations of section \eqref{RWE} can be explained under the light of Theorem \ref{proximityT}. It is more tractable for simplicity, to consider the estimate II given in \eqref{proxest3}-\eqref{proxest4}. The explicit linear growth rate $\alpha$ \eqref{proxest4} and its dependence on $\gamma>0$, $\delta<0$, $N$, suggests a moderate linear growth for short time intervals when $|\delta|$ and $|\gamma|$ are small, and $N$ is physically relevantly large. In particular, we observe that  the spatially averaged distance (which is equivalent and close to the  $\ell^{\infty}$ metric in the case of the finite lattices as it is justified from \eqref{equi})
\begin{equation}
	\label{spdist}
	D_a(t):=\frac{1}{\sqrt{Nh}}\|Y(t)\|_{\ell^2_{per}},
\end{equation}
has a linear growth rate, at most of $\mathcal{O}(1)$ for short time intervals, 
if 
\begin{equation}
	\label{small2}
	\gamma^3<\min\Big\{-\frac{\delta}{Nh},-\frac{\delta^3}{(\delta^2+1)hN^3}\Big\}\approx \mathcal{O}(1). 
\end{equation} 
The smallness requirement \eqref{small2}  is relevant in  physical setups where $\gamma$ and $-\delta$ are small, as chosen for our numerical experiments. It is highlighted in the top row of Figure \ref{Fig10}, depicting plots of the function \eqref{spdist} when $t\in [0,10]$ for the initial conditions \eqref{loc_ic} (panel (a)) and  \eqref{expini} (panel (b)). The rest of parameters are selected as in Figures \ref{Fig9A} and \ref{Fig9B}. In accordance with the theoretical estimates,  for these set of lattice parameters and time interval, $D_a(t)$  is smaller than $\mathcal{O}(1)$; we remark the almost linear growth after the minimum (attained close to the rogue wave event), which is still less than of $\mathcal{O}(1)$, for this time interval.

Moreover, the estimates \eqref{proxest3}-\eqref{proxest4} suggest that the distance between the structures for the above set of parameters, should be smaller when focused to the core of the MI pattern around $x_0=0$ for small time intervals. This prediction is  illustrated in the bottom row of Figure \ref{Fig10}, depicting  the time evolution of the  averaged distance 
\begin{equation}\label{Dr2}
	D_{a,r}:=  \frac{1}{\sqrt{hN_r}}\Big(h\sum_{-10}^{10}|u_n(t)- \phi_n(t)|^2\Big)^\frac{1}{2},
\end{equation}
where $N_r $ denotes the number of nodes in $x_n\in[-10,10]$.

\begin{figure}[!ht]
	\centering 
	\begin{tabular}{cc}
		(a)&(b)\\
			\includegraphics[scale=0.3]{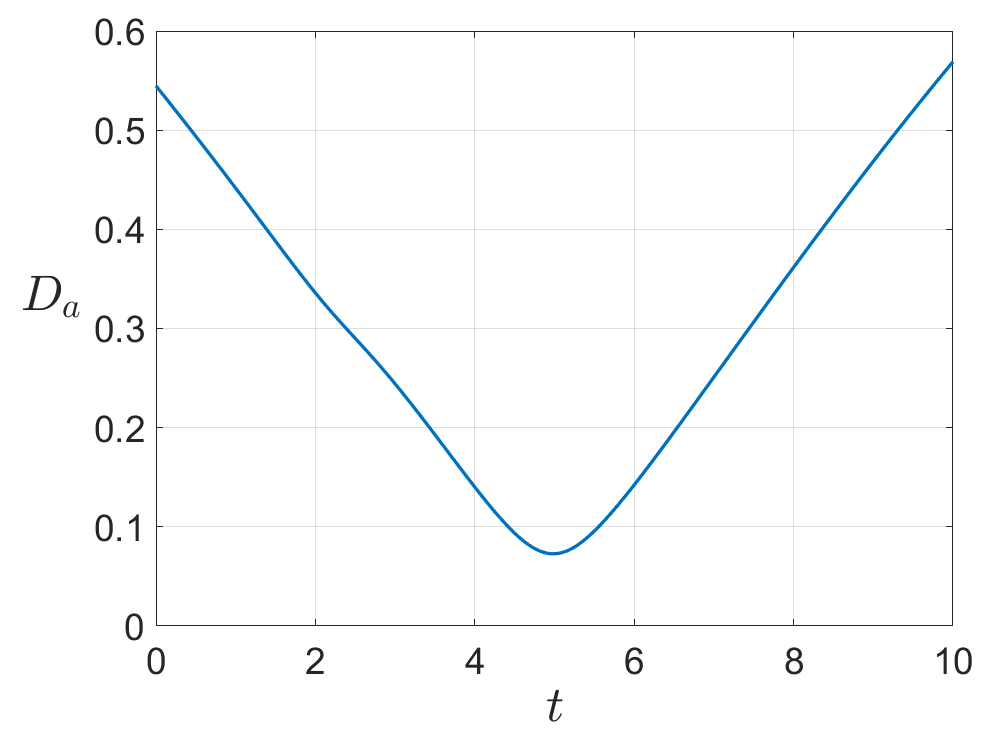}& 
					\includegraphics[scale=0.3]{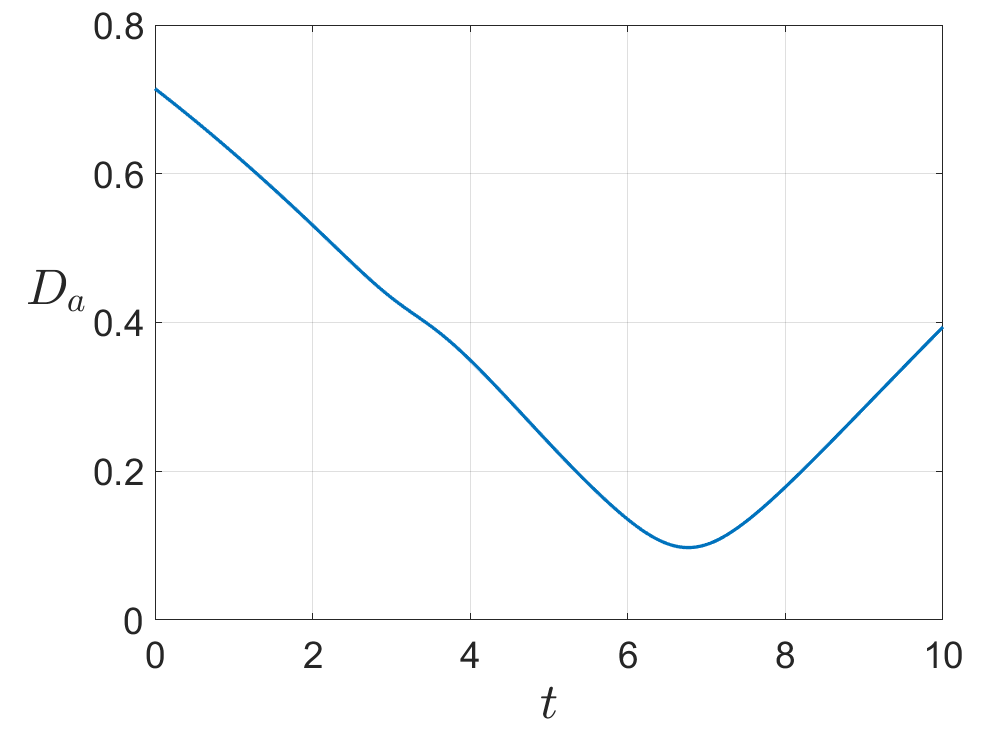}\\
				(c)&(d)\\
	\includegraphics[scale=0.3]{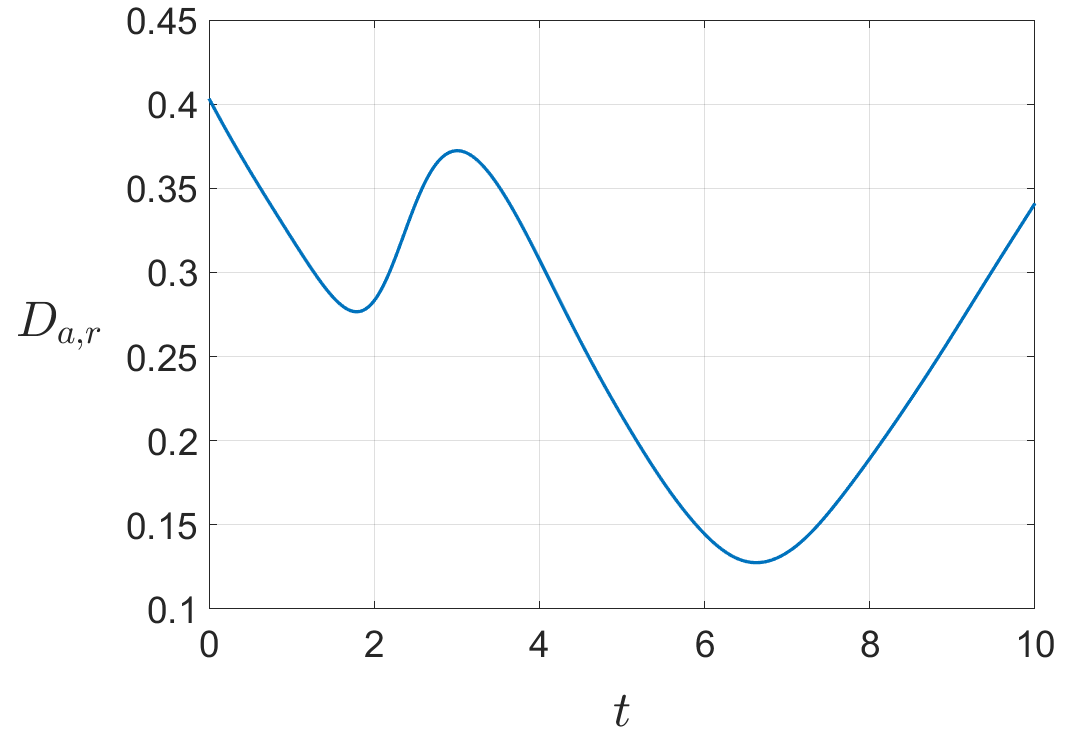}&
		\includegraphics[scale=0.3]{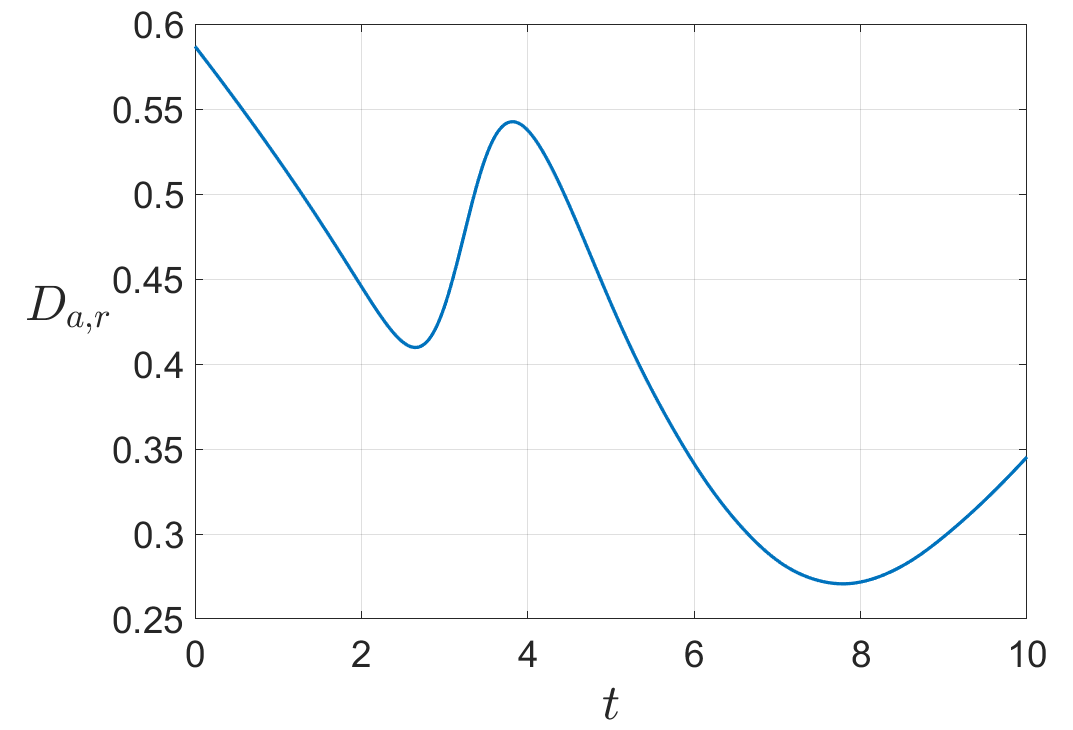}
	\end{tabular}
	\caption{Parameters as in Figures \ref{Fig9A} and  \ref{Fig9B}. Panels (a) and (b): Plot of the spatially averaged distance $D_a(t)$ \eqref{spdist} when $t\in [0,10]$,
			for the initial conditions \eqref{loc_ic} and  \eqref{expini}, respectively. Panels (c) and (d): Plot of $D_{a,r} (t)$ \eqref{Dr2} when $t\in [0,10]$ and  $x_n\in[-10,10]$, for the same initial conditions.}
	\label{Fig10}
\end{figure}

\subsection{Concluding remarks for the boundary conditions}
\begin{enumerate}
	\item
{\em Generalized boundary conditions at infinity: existence and uniqueness}. Motivated again by \cite{blmt2018} and \cite{bm2017}, instead of the boundary conditions \eqref{nv}, one may seek to consider more general nonzero boundary conditions at infinity, of the form 
\begin{equation}\label{nv2}
	\lim_{|n|\rightarrow\infty} u_n(t)=\lim_{|n|\rightarrow\infty}e^{i G^2t}\zeta_n,
\end{equation}
where $G\in\mathbb{R}$ and
\begin{equation}
	\label{nbc2}
	\zeta_n = \left\{
	\begin{array}{ll}
		\zeta_-, &n<0,
		\\
		\zeta_+, &n\geq 0,
	\end{array}
	\right.\;\;\;\mbox{$\zeta_{\pm}\in\mathbb{C}$ with $|\zeta_{\pm}|=\zeta$.}
\end{equation}
Using first the change of variables $u_n(t)=\psi_n(t)\exp(\rmi G^2t)$, and then,  the change of variables $U_n={\psi}_n-\zeta_n$,	
we derive that the modified system for $U_n$ reads now as 
\begin{align}
	\label{dnls_glb2}
	\mathrm{i}\dot{U}_n+ \Delta_d U_n+ \Delta_d \zeta_n-G^2(U_n+\zeta_n) &+|U_n+\zeta_n|^2(U_n+\zeta_n)\nonumber\\ &=\mathrm{i}\gamma (U_n+\zeta_n)+\mathrm{i}\delta|U_n+\zeta_n|^2(U_n+\zeta_n),\;\;
\end{align}
and the boundary conditions for $U_n$  are again \textit{zero}  at infinity, that is, the solution of \eqref{dnls_glb2}  must satisfy \eqref{Vbc1}. \emph{Note that the term $\Delta_d\zeta_n\in\ell^2$ by the definition of $\zeta_n$ in \eqref{nbc2}}.  The operators $\mathcal{F}_1$ and $\mathcal{F}_2$ are defined in this case by 
\begin{align*}
	\mathcal{F}_1(U_n)&=-G^2(U_n+\zeta_n)+|U_n+\zeta_n|^2(U_n+\zeta_n),\\
	\mathcal{F}_2(U_n)&=\mathrm{i}\gamma (U_n+\zeta_n)+\mathrm{i}\delta|U_n+\zeta_n|^2(U_n+\zeta_n).
\end{align*}
The operator $\mathcal{F}_1$ can be written as 
\begin{equation*}
	\mathcal{F}_1(U_n)=-G^2U_n-G^2\zeta_n+|\zeta_n|^2\zeta_n+\left\{|U_n|^2 U_n+2(|U_n|^2\zeta_n+U_n|\zeta_n|^2)+U_n^2\overline{\zeta}_n+\overline{U}_n\zeta_n^2\right\}, 
\end{equation*}
while the operator $\mathcal{F}_2$, as 
\begin{equation*}
	\mathcal{F}_2(U_n)=\mathrm{i}\gamma U_n+\mathrm{i}\gamma\zeta_n+\mathrm{i}\delta |\zeta_n|^2\zeta_n+\mathrm{i}\delta\left\{|U_n|^2 U_n+2(|U_n|^2\zeta_n+U_n|\zeta_n|^2)+U_n^2\overline{\zeta}_n+\overline{U}_n\zeta_n^2\right\}.
\end{equation*}
Carrying out the same analysis as for the proof of Theorem \ref{TH1} for the operators $\mathcal{F}_1$ and $\mathcal{F}_2$, we have that a unique solution $U\in C^1([0,T],\ell^p)$ of the modified equation \eqref{dnls_glb2} exists, if and only if 
\begin{equation*}
	-G^2\zeta_n+|\zeta_n|^2\zeta_n=0 \quad \text{and}\quad \mathrm{i}\gamma\zeta_n+\mathrm{i}\delta |\zeta_n|^2\zeta_n=0.
\end{equation*}
Therefore, we conclude with the following corollary.
\begin{corollary}
	\label{corol1}
	Consider the DNLS system \eqref{dnls_gl} with $\gamma>0,\delta<0$ in the infinite lattice supplemented with the nonzero boundary conditions \eqref{nv2}-\eqref{nbc2}. We also assume that the initial condition satisfies the boundary conditions \eqref{nv2}\eqref{nbc2}. Then the system has a solution if and only if $G^2=\zeta^2$ and $\zeta=A_*$. The solution is unique. 	
\end{corollary}
	\item {\em Dirichlet boundary conditions}.   (a) As proven in Theorem \ref{TH1}, Dirichlet boundary conditions are relevant for approximating problem \eqref{dnls_gl}-\eqref{nv} if and only if one uses the equivalent modified DNLS equation \eqref{dnls_glb} supplemented with zero Dirichlet boundary conditions. This approach approximates the zero boundary conditions at infinity \eqref{Vbc1} for \eqref{dnls_glb}. Repeating the numerical studies of sections IV and V with this approximation and using the change of variables $u_n(t)=(U_n(t)+A)\exp(\rmi A^2t)$ produces exactly the same results as those reported in these sections for the DNLS \eqref{dnls_gl} supplemented with periodic boundary conditions.
	(b) Clearly, Dirichlet boundary conditions for the original, unmodified DNLS \eqref{dnls_gl} are not relevant for the finite lattice approximation of the non-zero boundary conditions at infinity \eqref{nv}.
	\item {\em Generalized boundary conditions: asymptotic behavior}. Corollary \ref{corol1} generalizes Theorem \ref{TH1} to the case of more general nonzero boundary conditions at infinity \eqref{nv2}-\eqref{nbc2}. Then, one can speculate that the results on the structure of the attractor of section III and the numerical results of sections IV and V can be of similar nature as for the simplest case of boundary conditions \eqref{nv}. However, this speculation, to be verified, needs theoretical and numerical analysis either (i) on the modified DNLS \eqref{dnls_glb2} supplemented with Dirichlet boundary conditions or (ii) on the unmodified DNLS \eqref{dnls_gl} supplemented, this time, with the relevant non-zero Dirichlet boundary conditions. Such an analysis may have similarities but also differences compared with the one of the boundary conditions \eqref{nv} and may deserve independent attention to be considered elsewhere. In this regard, even more general boundary conditions can be considered with a varying $\zeta_n$, such that $\lim_{n\rightarrow\infty} |\zeta_n|=\zeta$ (a discretized version of the boundary conditions of \cite[Section 3, pg. 127]{JDE2024}).
\end{enumerate}

\section{Conclusions}
\label{conclusions}
In this work, we studied the dynamics of the Discrete Nonlinear Schr\"odinger Equation incorporating linear gain and nonlinear loss effects, a  significant model in various physical contexts. In the case of the infinite lattice we considered the problem with nonzero boundary conditions proving that solutions exist if and only if the supporting background has a prescribed critical value $A_*$ defined by the gain and loss strengths. We argued that this property cannot be captured by finite lattice approximations, as those defined by the system supplemented with periodic boundary conditions, which are essential for the numerical simulations of the problem. For the periodic lattice, which is of physical and mathematical significance itself,  we proved that the dynamics are dominated by the convergence to a global attractor which is the unique plane wave of constant amplitude, namely $A_*$, with a novel argument which studies not only the limits of the amplitudes of the initial conditions but also their frequencies and spectrum.  The convergence arguments accompanied by the analysis for modulation instability justify that this fundamental mechanism for the emergence of localized structures can be only transient in the case of the periodic lattice. The dynamics of localized initial data are investigated numerically: in the case where $A=A_*$ which is the only relevant in order to approximate the infinite lattice, we identify the persistence of the major characteristics of the modulation instability patterns manifested in the dynamics of integrable NLS systems, as the Ablowitz-Ladik lattice: the oscillatory  wedge shaped sector and the two outer plane wave regions of the undisturbed background of amplitude $A_*$, which is  globally asymptotically stable in the case of the finite lattice. In the case $A\neq A_*$ which is only relevant for the finite lattice, the instability pattern still persists, this time on a varying background transiently, prior the convergence of the dynamics to the plane wave attractor. We may conclude that the analysis of the infinite lattice ``leaves its mark" on the dynamics of the finite-dimensional one, particularly regarding the manifestation of the  major characteristics  of the universal behavior in the nonlinear stage of modulation instability. 

At the early stages of the dynamics we identified the emergence of rogue wave events which are remarkably proximal to the analytical discrete Peregrine rogue wave solution of the Ablowitz-Ladik lattice. This proximity is justified analytically by suitable estimates for the distance between the solutions of the DNLS system and the Ablowitz-Ladik lattice. This proximity, in light of the theoretical analysis for the infinite lattice, may pave the way for new and interesting investigations for the potential construction of rogue wave-like solutions for the dissipative DNLS via fixed-point iterations, using the  numerical methods outlined in \cite{SDP2018,EP,WM}.

As a summary, we argue that the breaking of integrability by the considered dissipative DNLS system is not dramatic since the major features of the integrable dynamics are preserved. The above results showcase that the behavior due to modulation instability can be  universal and observed not only in Hamiltonian NLS systems but also in even more realistic set-ups described by non integrable systems where dissipation and forcing effects are present.  
\begin{acknowledgments}
	The authors G.F. and N.I.K.  acknowledge that this work was partially supported by the Xiamen University Malaysia Research Fund (Grant No: XMUMRF/2022-C9/IMAT/0020).
\end{acknowledgments}
\vspace{2cm}
\noindent
\textbf{Authors Declarations}\\
The authors have no conflicts to disclose.\\
\\
\textbf{Authors Contributions Statement}\\
All authors contributed equally to the study conception, design and writing of the manuscript. Material preparation, data collection and analysis were performed equally by all authors.  All authors read and approved the final manuscript.

\end{document}